\documentclass{article} \usepackage[margin=1in]{geometry}
\usepackage[T1]{fontenc}

\usepackage{hyperref}
\usepackage{url}

\usepackage{amsmath,amssymb,amsthm}
\usepackage{cleveref}
\usepackage{xspace}
\usepackage[font=small]{caption}
\usepackage{subcaption}
\usepackage{boxedminipage}
\usepackage{hhline}
\usepackage{multicol}
\usepackage{multirow}
\usepackage{booktabs}
\usepackage{centernot}
\usepackage{algorithm}
\usepackage{algpseudocode}
\usepackage{appendix}
\usepackage{color,graphicx}
\usepackage{paralist}

\theoremstyle{plain}
\newtheorem{theorem}{Theorem}[section]
\newtheorem{lemma}[theorem]{Lemma}

\newtheorem{fact}[theorem]{Fact}

\theoremstyle{definition}
\newtheorem{definition}[theorem]{Definition}

\newtheorem{remark}[theorem]{Remark}

\DeclareMathOperator{\poly}{poly}
\DeclareMathOperator{\cost}{cost}
\DeclareMathOperator{\dist}{dist}
\DeclareMathOperator{\Ball}{Ball}
\DeclareMathOperator{\Balls}{Balls}
\DeclareMathOperator{\opt}{opt}
\DeclareMathOperator{\ring}{ring}

\newcommand{\ProblemName}[1]{\textsc{#1}}
\newcommand{\kMedian}{\ProblemName{$k$-Median}\xspace}
\newcommand{\kMeans}{\ProblemName{$k$-Means}\xspace}
\newcommand{\kzC}{\ProblemName{$(k, z)$-Clustering}\xspace}
\newcommand{\kmMedian}{\ProblemName{$(k,m)$-Robust Median}\xspace}
\newcommand{\onemMedian}{\ProblemName{$(1,m)$-Robust Median}\xspace}
\newcommand{\kmMeans}{\ProblemName{$(k,m)$-Robust Means}\xspace}
\newcommand{\kzmC}{\ProblemName{$(k, z, m)$-Robust Clustering}\xspace}
\newcommand{\kztC}{\ProblemName{$(k, z, t)$-Robust Clustering}\xspace}
\newcommand{\kMeanspp}{\ProblemName{$k$-Means++}\xspace}

\newcommand{\pclose}{\ensuremath{p_{\mathrm{close}}\xspace}}
\newcommand{\pfar}{\ensuremath{p_{\mathrm{far}}\xspace}}
\newcommand{\Tclose}{\ensuremath{T_{\mathrm{close}}\xspace}}
\newcommand{\Tfar}{\ensuremath{T_{\mathrm{far}}\xspace}}
\newcommand{\Tclosez}{\ensuremath{T^z_{\mathrm{close}}\xspace}}
\newcommand{\Tfarz}{\ensuremath{T^z_{\mathrm{far}}\xspace}}
\newcommand{\Cclose}{\ensuremath{C_{\mathrm{close}}\xspace}}
\newcommand{\Cfar}{\ensuremath{C_{\mathrm{far}}\xspace}}
\newcommand{\Gall}{\ensuremath{G_{\mathrm{all}}\xspace}}
\newcommand{\Rall}{\ensuremath{R_{\mathrm{all}}\xspace}}
\algrenewcommand\algorithmicrequire{\textbf{Input:}}
\algrenewcommand\algorithmicensure{\textbf{Output:}}

\usepackage{amsmath,amsfonts,bm}

\def\eqref#1{equation~\ref{#1}}

\def\1{\bm{1}}

\def\eps{{\epsilon}}

\DeclareMathAlphabet{\mathsfit}{\encodingdefault}{\sfdefault}{m}{sl}
\SetMathAlphabet{\mathsfit}{bold}{\encodingdefault}{\sfdefault}{bx}{n}

\def\gG{{\mathcal{G}}}

\def\gL{{\mathcal{L}}}

\def\gR{{\mathcal{R}}}

\newcommand{\R}{\mathbb{R}}

\DeclareMathOperator*{\argmin}{arg\,min}

\title{Near-optimal Coresets for Robust Clustering}

\author{Lingxiao Huang\thanks{
    Email: \texttt{huanglingxiao1990@126.com}
  }\\
    Huawei TCS Lab
    \and
    Shaofeng H.-C. Jiang\thanks{Research partially supported by a national key R\&D program of China No.\ 2021YFA1000900,
    a startup fund from Peking University, and the Advanced Institute of Information Technology, Peking University.
    Email: \texttt{shaofeng.jiang@pku.edu.cn}
  }\\
  Peking University
  \and
  Jianing Lou\thanks{ 
    Email: \texttt{loujn@pku.edu.cn}
  }\\
  Peking University
  \and
  Xuan Wu\thanks{ 
Email: \texttt{wu3412790@gmail.com}}\\
  Huawei TCS Lab
}

\begin{document}

\maketitle

\begin{abstract}
    
    We consider robust clustering problems in $\mathbb{R}^d$,
    specifically $k$-clustering problems (e.g., \kMedian and \kMeans) with $m$ \emph{outliers},
    where the cost for a given center set $C \subset \mathbb{R}^d$ aggregates the distances from $C$ to all but the furthest $m$ data points,
    instead of all points as in classical clustering.
We focus on the $\eps$-coreset for robust clustering, a small proxy of the dataset that preserves the clustering cost within $\eps$-relative error for all center sets.
Our main result is an $\eps$-coreset of size $O(m + \poly(k \eps^{-1}))$ that can be constructed in near-linear time.
This significantly improves previous results,
which either suffers an exponential dependence on $(m + k)$~\cite{feldman2012data}, or
has a weaker bi-criteria guarantee~\cite{huang2018epsilon}.
Furthermore, we show this dependence in $m$ is nearly-optimal,
and the fact that it is isolated from other factors
may be crucial for dealing with large number of outliers.
We construct our coresets by adapting to the outlier setting a recent framework~\cite{braverman2022power} 
which was designed for capacity-constrained clustering,
overcoming a new challenge that the participating terms in the cost,
particularly the excluded $m$ outlier points, are dependent on the center set $C$.
We validate our coresets on various datasets, and we observe a superior size-accuracy tradeoff
compared with popular baselines including uniform sampling and sensitivity sampling.
We also achieve a significant speedup of existing approximation algorithms for robust clustering using our coresets.
\end{abstract}

\section{Introduction}
\label{sec:intro}

We give near-optimal $\eps$-coresets for \kMedian and \kMeans (and more generally, \kzC)
with outliers in Euclidean spaces.
Clustering is a central task in data analysis,
and popular center-based clustering methods, such as \kMedian and \kMeans, have been widely applied.
In the vanilla version of these clustering problems,
given a \emph{center set} of $k$ points $C$,
the objective is usually defined by the sum of (squared) distances from each data point to $C$.

This formulation, while quite intuitive and simple to use,
has severe robustness issues when dealing with noisy/adversarial data;
for instance, an adversary may add few noisy \emph{outlier} points that are far from the center
to ``fool'' the clustering algorithm to wrongly put centers towards those points in order to minimize the cost.
Indeed, such robustness issue introduced by outliers has become
a major challenge in data science and machine learning,
and it attracted extensive algorithmic research on the topic~\cite{charikar2001algorithms,chen2008constant,Cands2011RobustPC,chawla2013kmeans,mount2014on,gupta2017local,statman2020kmeans,DBLP:conf/icml/DingW20}.
Moreover, similar issues have also been studied from the angle of statistics~\cite{huber2009robust}.

\paragraph{Robust Clutering}
We consider \emph{robust} versions of these clustering problems,
particularly a natural and popular variant, called clustering with outliers~\cite{charikar2001algorithms}.
Specifically, given a dataset $X\subset \R^d$, the \kzmC problem is to find a center set $C\subset \R^d$ of $k$ points (repetitions allowed), that minimizes the objective function
\begin{equation}
\label{eq:robust_cost}
	\cost^{(m)}_z(X,C) := \min_{L\subseteq X: |L|=m } \sum_{x\in X\setminus L} (\dist(x,C))^z.
\end{equation}
Here, $L$ denotes the set of \emph{outliers}, $\dist$ denotes the Euclidean distance, and $\dist(x,C) := \min_{c \in C} \dist(x, c)$.
Intuitively, the outliers capture the furthest points in a cluster which are ``not well-clustered'' and are most likely to be the noise.
Notice that the parameter $z$ captures various (robust) clustering problems, including \kmMedian (where $z=1$), \kmMeans (where $z=2$).
On the other hand, if the number of outliers $m = 0$ then the robust clustering problem falls back to the non-robust version.
The \kzmC problem has been widely studied in the literature~\cite{chen2008constant,gupta2017local,RLS18,friggstad2019approximation,statman2020kmeans}.
Moreover, the idea of removing outliers has been also considered in other machine learning tasks, e.g., robust PCA~\cite{bhaskara2018low} and robust regression~\cite{rousseeuw1987robust,mount2014on}.

\paragraph{Computational Challenges}
However, the presence of outliers introduces significant computational challenges,
and it inspires a series of research to design efficient algorithms for robust clustering.
On one hand, approximation algorithms with strict accuracy guarantee has been obtained
~\cite{charikar2001algorithms,chen2008constant,gupta2017local,RLS18,feng2019improved,friggstad2019approximation,zhang2021local}
but their running time is a high-degree polynomial which is impractical.
On the other hand, more scalable algorithms were also proposed~\cite{bhaskara2019greedy,deshpande2020robust},
however, the approximation ratio is worse, and a more severe limitation
is that their guarantee usually violates the required number of outliers.
Moreover, to the best of our knowledge, we are not aware of works
that design algorithms in sublinear models, such as streaming and distributed computing.

\paragraph{Coresets}
In order to tackle the computational challenges,
we consider \emph{coresets} for robust clustering.
Roughly, an $\eps$-coreset is a tiny proxy of the massive input dataset,
on which the clustering objective is preserved within $\eps$-error for every potential center set.
Existing algorithms may benefit a significant speedup if running on top of a coreset,
and more importantly, coresets can be used to derive sublinear algorithms,
including streaming algorithms~\cite{DBLP:conf/stoc/Har-PeledM04},
distributed algorithms~\cite{BalcanEL13} and dynamic algorithms~\cite{DBLP:conf/esa/HenzingerK20},
which are highly useful to deal with massive datasets.

Stemming from~\cite{DBLP:conf/stoc/Har-PeledM04}, the study of coresets for the non-robust version of clustering, i.e., \kzC,
has been very fruitful~\cite{feldman2011unified,FSS20,sohler2018strong,huang2020coresets,BJKW21,cohen2021new,braverman2022power}, and the state-of-the-art coreset achieves
a size $\poly(k \eps^{-1})$, independent of $d$ and $n$.
However, coresets for robust clustering were much less understood.
Existing results either suffers an exponential $(k + m)^{k + m}$ factor in the coreset size~\cite{feldman2012data},
or needs to violate the required number of outliers~\cite{huang2018epsilon}.
This gap leads to the following question: can we efficiently construct an $\eps$-coreset of size $\poly(m, k, \eps^{-1})$ for \kzmC (without violating the number of outliers)?

\subsection{Our contributions}
\label{sec:contribution}
Our main contribution, stated in \Cref{thm:intro_main},
is a near-optimal coreset for robust clustering,
affirmatively answering the above question.
In fact, we not only achieve $\poly(m)$,
but also linear in $m$ and is isolated from other factors.
This can be very useful when the number of outliers $m$ is large.
\begin{theorem}[Informal; see \Cref{thm:main}]
	\label{thm:intro_main}
	There exists a near-linear time algorithm that given data set $X\subset \R^d$, $z\geq 1$, $\eps \in (0,0.3)$ and integers $k,m\geq 1$,
	computes an $\eps$-coreset of $X$ for \kzmC of size $O(m)+2^{O(z\log z)}\tilde{O}(k^3\eps^{-3z-2})$, with constant probability.
\end{theorem}
Our coreset improves over previous results in several aspects.
Notably, compared with~\cite{feldman2012data},
our result avoids the exponential $(k + m)^{k + m}$ factor in the coreset size,
and our coreset has a strict guarantee for $m$ outliers instead of
a bi-criteria guarantee as in~\cite{huang2018epsilon} that needs to allow more or fewer outliers in the objective for the coreset.
We also note that our coreset is composable (\Cref{remark:thm}).

Furthermore, we show that the linear dependence in $m$ is not avoidable in coreset size (\Cref{thm:intro_lb}).
Hence, combining this with a recent size lower bound of $\Omega(k \eps^{-2})$~\cite{cohen2022towards}
for vanilla clustering (i.e., $m = 0$),
we conclude that the dependence of every parameter (i.e., $m, k, \eps$) in our coreset is nearly tight.

\begin{theorem}
	\label{thm:intro_lb}
	For every integer $m \geq 1$, there exists a dataset $X \subset \mathbb{R}$ of $n \geq m$ points,
	such that for every $0< \epsilon < 0.5$,
	any $\eps$-coreset for \onemMedian must have size $\Omega(m)$.
\end{theorem}
For the lower bound, we observe that when $m = n - 1$, the clustering cost for \onemMedian reduces to the distance to the nearest-neighbor
from the center $c$. This is easily shown to require $\Omega(n) = \Omega(m)$
points in the coreset, in order to achieve any finite approximation.
The formal proof can be found in \Cref{sec:lb}.

\paragraph{Experiments}
We evaluate the empirical performance of our coreset on various datasets (in \Cref{sec:exp}).
We validate the size-accuracy tradeoff of our coreset compared with popular coreset construction methods,
particularly uniform sampling (which is a natural heuristic)
and sensitivity sampling~\cite{feldman2012data},
and we observe that our coreset consistently outperforms these baselines in accuracy
by a significant margin for every experimented coreset size ranging from $500$ to $5000$.
We also run existing approximation algorithms on top of our coreset,
and we achieve about 100x speedup for both a) a Lloyd heuristic adopted to the outlier setting~\cite{chawla2013kmeans} that is seeded by an outlier-version of \kMeanspp~\cite{bhaskara2019greedy},
and b) a natural local search algorithm~\cite{friggstad2019approximation}.
These numbers show that our coreset is not only near-optimal in theory,
but also demonstrates the potential to be used in practice.

\subsection{Technical Overview}
\label{sec:overview}
Similar to many previous coreset constructions,
we first compute a near-optimal solution,
an $(\alpha,\beta,\gamma)$-approximation (see \Cref{def:abgapprox})
$C^* := \{c_i^*\mid i\in [\beta k]\}$, obtained using known approximation algorithms (see the discussion in \Cref{sec:tri_criteria}).
Then with respect to $C^*$,
we identify the outliers $L^*\subset X$ of $C^*$
and partition the remaining inlier points $X \setminus L^*$ into $|C^*|$ clusters $\{X_i\}_i$.

We start with including $L^*$ into our coreset,
and we also include a weighted subset of the remaining inlier points $X\setminus L^*$
by using a method built upon a recent framework~\cite{braverman2022power},
which was originally designed for clustering with capacity constraints.
The step of including $L^*$ in the coreset is natural,
since otherwise one may miss the remote outlier points which can incur a huge error;
furthermore, the necessity of this step is also justified by our $\Omega(m)$ lower bound (\Cref{thm:intro_lb}).
Similar to \cite{braverman2022power},
for each cluster $X_i$ among the remaining inliers points $X \setminus L^*$,
we identify a subset of $X_i$ that consists of $\poly(k\epsilon^{-1})$ \emph{rings},
and merge the remaining part into $\poly(k \epsilon^{-1})$ \emph{groups} of rings such that each group $G$ has a tiny cost (\Cref{thm:meta}).
We use a general strategy that is similar to \cite{braverman2022power} to
handle separately the rings (\Cref{lem:intro:uniform}) and groups (\Cref{lem:intro:twopoints}),
but the actual details differ significantly due to the presence of outliers.

\paragraph{Handling Rings}
For a ring data subset $R\subseteq X_i$, i.e., a subset such that every point is at a similar distance (up to a factor of $2$) to the center $c_i^*$,
we apply a uniform sampling on it to construct a coreset (with \emph{additive} error, \Cref{def:additivecoreset}), as in \cite{braverman2022power}.
However, our guarantee is stronger (\Cref{lem:intro:uniform}) in that
the robust clustering cost in $R$ needs to be preserved for all the possible number of outliers $0 \leq t \leq m$ \emph{simultaneously}.
Indeed, we need this stronger guarantee since the number of outliers that lie in $R$
is not known a priori and can be any number between $0$ and $m$.
To prove this guarantee, we rewrite the robust clustering cost as an integration of ball ranges
(Fact~\ref{fact:integral}) and use a fact that uniform sampling approximately estimates all ball ranges (\Cref{lem:approximation}).
Similar idea of writing the cost as an integration has also been used in previous works, e.g., \cite{huang2018epsilon,braverman2022power}.

\paragraph{Handling Groups}
The main technical difficulty is the handling of groups (\Cref{lem:intro:twopoints}).
We still construct a two-point coreset (\Cref{def:twopoints}) for every group $G\subset X_i$, as in~\cite{braverman2022power}.
To analyze the error of this two-point coreset for an arbitrary center set $C \subset \mathbb{R}^d$,
we partition the groups into \emph{colored} and \emph{uncolored} groups \emph{with respect to $C$} (\Cref{lem:color}) in a way similar to \cite{braverman2022power}.
The error analysis for the colored group is still similar, where we charge the error to the additive term (recalling that it suffices to find a coreset for a ring with additive error),
except that we need to modify the technical details to handle the outliers.
For the uncolored groups,
we observe that if an uncolored group $G$ contains no outlier or all points in $G$ are outliers,
then the two-point coreset in every uncolored group $G$ only incurs multiplicative error (\Cref{lem:uncol}).
This observation generalizes a similar observation in~\cite{braverman2022power}.
However, this observation may not hold when $G$ only \emph{partially} intersect the outliers.
Luckily, we manage to prove that there are at most \emph{two} uncolored groups
that partially intersect the outliers, for every center set $C$ (\Cref{lem:special}).
This suffices for bounding the total error incurred by the uncolored groups.
Here, the key geometric observation is that the outliers must lie in
\emph{consecutive} uncolored groups due to the way we decompose $X_i$,
and apart from the two groups that partially intersect the outliers,
every other group within the consecutive sequence consists of outliers only.

\subsection{Other Related Work}
\label{sec:related}

\paragraph{Robust Clustering in $\R^d$}
Robust clustering, first proposed by~\cite{charikar2001algorithms}, has been studied for two decades.
For \kmMedian, \cite{charikar2001algorithms} designed a bi-criteria approximate algorithm with violations on $k$.
\cite{chen2008constant} first showed a pure constant approximate algorithm, whose approximate ratio was improved to $7.081+\eps$~\cite{RLS18}.
When $k=O(1)$, \cite{feng2019improved} also proposed a PTAS for \kmMedian.
For \kmMeans, \cite{gupta2017local} designed a bi-criteria approximate algorithm with violations on $m$.
\cite{RLS18} first proposed a constant approximate algorithm, and the approximate ratio was improved to $6+\eps$~\cite{feng2019improved}.
For general \kzmC, \cite{friggstad2019approximation} achieved an $O(z^z)$ approximate solution with $(1+\eps)k$ centers.
Due to the wide applications, scalable algorithms have been designed for \kmMeans~\cite{bhaskara2019greedy,deshpande2020robust} besides theoretical study, which may have a worse provable guarantee but are more efficient in practice.

\paragraph{Coresets for Clustering}
There is a large body of work that studies coreset construction for vanilla \kzC in $\R^d$~\cite{DBLP:conf/stoc/Har-PeledM04,feldman2011unified,braverman2016new,huang2018epsilon,cohen2021new,cohen2022towards}.
The state-of-art result for general \kzC is by~\cite{cohen2022towards}, where the coreset size is $\tilde{O}(z^{O(z)} k\eps^{-2}\cdot \min\left\{k,\eps^{-z}\right\})$.
This bound nearly matches a lower bound of $\Omega(k\eps^{-2}+k\min\left\{d,2^{z/20}\right\})$~\cite{cohen2022towards,huang2020coresets}.
In addition, coresets for constrained clustering in Euclidean spaces
has also been considered, such as capacitated clustering and the tightly related
fair clustering~\cite{schmidt2019fair,huang2019coresets,braverman2022power},
and ordered weighted clustering~\cite{braverman2019coresets}.
Going beyond Euclidean spaces, coresets of size $\poly(k \eps^{-1})$
were known for \kzC in doubling metrics~\cite{huang2018epsilon},
shortest-path metrics of graphs with bounded treewidth~\cite{baker2020coresets}
and graphs that exclude a fixed minor~\cite{BJKW21}.

 \section{Preliminaries}
\label{sec:prelim}

\paragraph{Balls and Rings} For a point $a\in \mathbb{R}^d$, and positive real numbers $r'>r>0$, define $\Ball(a,r)=\{x\in\mathbb{R}^d,\dist(x,a)\leq r\}$ and $\ring(a,r,r')=\Ball(a,r')\setminus \Ball(a,r)$.
For a set of points $A\subset\mathbb{R}^d$, $\mathrm{Balls}(A,r)=\cup_{a\in A} \mathrm{Ball}(a,r)$.

\paragraph{Weighted Outliers}
Since our coreset uses weighted points, we need to define the notion of weighted sets and weighted outliers.
We call a set $S$ with an associated weight function $w_S:S\rightarrow \R_{\geq 0}$ a \emph{weighted set}.
Given two weighted sets $(X,w_X)$ and $(Y,w_Y)$ such that $Y\subseteq X$
and $w_Y(x)\leq w_X(x)$ for any $x\in Y$,
let $X-Y$ denote a weighted set $(Z,w_Z)$ such that $w_Z=w_X-w_Y$,\footnote{Here, if $x\notin Y$, we let $w_Y(x)=0$.} and $Z$ is the support of $w_Z$.
Moreover, for a weighted set $X$,
we denote $\gL_X^{(m)}$ as the collection of all possible sets of \emph{weighted outliers} $(Y,w_Y)$ satisfying that
$Y\subseteq X$, $\sum_{x\in Y} w_Y(x)=m$ and that $\forall x\in X$, $w_Y(x)\leq w_X(x)$.
In this definition, since $X$ is a weighted point set, we need to pick outliers of \emph{total weight} $m$ in the objective $\cost_z^{(m)}(X,C)$,
instead of $m$ distinct points which may have a much larger weights than $m$.

\paragraph{Weighted Cost Functions}
For $m = 0$, we write $\cost_z$ for $\cost^{(m)}_z$.
We extend the definition of the cost function to that on a weighted set $X \subset \mathbb{R}^d$.
For $m = 0$, we define $\cost_z(X, C) := \sum_{x \in X} w_X(x) \cdot (\dist(x, C))^z$.
For general $m \geq 1$, the cost is defined using the notion of weighted outliers and aggregating using the $\cost_z$ function which is the $m = 0$ case.
\[
	\cost_z^{(m)}(X, C) := \min_{ (L, w_l) \in \gL_X^{(m)} }\left\{
		\cost_z(X - L, C)
	\right\}.
\]
One can check that this definition is a generalization of the unweighted case (\ref{eq:robust_cost}).
For a weighted set $X \subset \mathbb{R}^d$,
let the optimal solution be $\opt_z^{(m)}(X):=\min_{C \subset \mathbb{R}^d, |C|=k} \cost_z^{(m)}(X,C)$.

\begin{definition}[Coreset]
	\label{def:coreset}
	Given a point set $X\subset \R^d$ and $\eps\in (0,1)$, an $\eps$-coreset for \kzmC is a weighted subset $(S,w_{S})$ of $X$ such that
\begin{equation}
		\label{eq:robust_coreset_cost}
		\forall C \subset \mathbb{R}^d, |C| =k, \qquad \cost_z^{(m)}(S,C) \in (1\pm \eps)\cdot \cost_z^{(m)}(X,C)
	\end{equation}
\end{definition}

Even though \Cref{def:coreset} naturally extends the definition
of coresets for vanilla clustering~\cite{DBLP:conf/stoc/Har-PeledM04,feldman2011unified,FSS20},
it is surprising that this exact definition did not seem to appear in the literature.
A closely related definition~\cite{huang2018epsilon}
considers a relaxed ``bi-criteria'' (with respect to the number of outliers) guarantee of the cost, i.e.,
	$(1-\eps)\cdot \cost_z^{(1+\beta)m}(S,C)\leq \cost_z^{(m)}(X,C)\leq (1+\eps)\cdot \cost_z^{(1-\beta)m}(S,C)$,
for $\beta\in [0,1)$,
and their coreset size depends on $\beta^{-1}$.
Another definition was considered in~\cite{feldman2012data}, 
which considers a more general problem called weighted clustering (so their coreset implies our \Cref{def:coreset}).
Unfortunately, this generality leads to an exponential-size coreset (in $k, m$).

\begin{definition}[$(\alpha,\beta,\gamma)$-Approximation] \label{def:abgapprox}
 Given a dataset $X\subset \R^d$ and real numbers $\alpha,\beta,\gamma\geq 1$, an $(\alpha,\beta,\gamma)$-approximate solution
for \kzmC on $X$ is a center set $C^*\subset \mathbb{R}^d$ with $|C^*|\leq \beta k$ such that
$\cost_z^{(\gamma m)}(X,C^*)\leq \alpha \cdot \opt_z^{(m)}(X)$.
\end{definition}

\begin{lemma}[Generalized triangle inequalities]
\label{lem:gentri}
Let $a,b\geq 0$ and $\delta\in (0,1)$, then for $z\geq 1$,
\begin{enumerate}
\item (Lemma A.1 of \cite{MMR19}) $(a+b)^z\leq (1+\delta)^{z-1}\cdot a^z+(1+\frac{1}{\delta})^{z-1}\cdot b^z$
\item (Claim 5 of \cite{sohler2018strong}) $(a+b)^z\leq (1+\delta)\cdot a^z+(\frac{3z}{\delta})^{z-1}\cdot b^z$
\end{enumerate}
\end{lemma} 
 \section{Coresets for \kzmC}
\label{sec:main}

We present our main theorem in \Cref{thm:main}.
As mentioned, the proof of \Cref{thm:main} is based on the framework in~\cite{braverman2022power},
and we review the necessary ingredients in \Cref{sec:review}.
The statement of our algorithm and the proof of \Cref{thm:main} can be found in \Cref{sec:proof_main}.

\begin{theorem} \label{thm:main} Given input dataset $P\subset \mathbb{R}^d$ with $|P|=n$, integers $k,m\geq 1$, and real number $z\geq 1$ and assume there exists an algorithm that computes an $(\alpha,\beta,\gamma)$-approximation of $X$ for \kzmC in time $\mathcal{A}(n,k,d,z)$, then \Cref{alg:main} uses time $\mathcal{A}(n,k,d,z)+O(nkd)$ to construct a weighted subset $(S,w_S)$ with size $|S|=\gamma m+2^{O(z\log z)}\cdot \beta\cdot \tilde{O}(k^3\epsilon^{-3z-2})$,
such that with probability at least 0.9, for every integer $0\leq t\leq m$, $S$ is an $\alpha \epsilon$-coreset $S$ of $X$ for \kztC.
\end{theorem}

\begin{remark} 
	\label{remark:thm}
	By rescaling $\epsilon$ to $\epsilon/\alpha$ in the input of \Cref{alg:main}, we obtain an $\epsilon$-coreset of size $\gamma m+2^{O(z\log z)}\cdot \alpha^{3z+2}\beta\cdot \tilde{O}(k^3\epsilon^{-3z-2})$.
We discuss how to obtain $(\alpha, \beta, \gamma)$-approximations in \Cref{sec:tri_criteria}.
We also note that \Cref{thm:main} actually yields an $\epsilon$-coreset for $\kztC$ \emph{simultaneously} for every integer $0\leq t\leq m$, which implies that our coreset is \emph{composable}. 
Specifically, if for every integer $0\leq t\leq m$, $S_X$ is an $\epsilon$-coreset of $X$ for $\kztC$ and $S_Y$ is an $\epsilon$-coreset of $X$ for $\kztC$, then for every integer $0\leq t\leq m$, $S_X\cup S_Y$ is an $\epsilon$-coreset of $X\cup Y$ for \kztC.
\end{remark}

\subsection{The Framework of \cite{braverman2022power}} \label{sec:review}

\Cref{thm:meta} is a general geometric decomposition theorem for coresets
which we use crucially.
It partitions an arbitrary cluster into $\poly(k/\epsilon)$ \emph{rings}
and merge the remaining rings into $\poly(k/\epsilon)$ \emph{groups} with low contribution to $\cost_z(X_i,c_i^*)$. (See \Cref{fig:ring_group} for an illustration.)
\begin{theorem}[Decomposition into rings and groups~{\cite[Theorem 3.2]{braverman2022power}}]
\label{thm:meta}
Let $X\subset \mathbb{R}^d$ be a set and $c\in \mathbb{R}^d$ be a center point.
There exists an $O(nkd)$-time algorithm that computes a partition of $X$ into two disjoint collections of sets $\mathcal{R}$ and $\mathcal{G}$,
such that $X=(\cup_{R\in \mathcal{R}} R)\cup (\cup_{G\in \mathcal{G}} G)$, where $\mathcal{R}$ is a collection of disjoint rings satisfying
\begin{enumerate}
	\item $\forall R\in \gR$, $R$ is a ring of the form $R=R_i(X,c)$ for some integer $i\in \mathbb{Z}\cup \{-\infty\}$, where $R_i(X,c):=X\cap \ring(c,2^{i-1},2^i)$ for $i\in \mathbb{Z}$ and $R_{-\infty}(X,c):=X\cap \{c\}$
	\item $|\gR|\leq 2^{O(z\log z)}\cdot \tilde{O}(k\epsilon^{-z})$
\end{enumerate}
and $\mathcal{G}$ is a collection of disjoint groups satisfying
\begin{enumerate}
\item $\forall G\in \mathcal{G}$, $G$ is the union of consecutive rings of $(X,c)$.
Formally, $\forall G\in \mathcal{G}$, there exists two integers $-\infty\leq l_G\leq r_G$ such that $G=\cup_{i=l_G}^{r_G} R_i(X,c)$ and the intervals $\{[l_G,r_G],G\in\mathcal{G}\}$ are disjoint for different $G\in \gG$
\item $|\mathcal{G}|\leq 2^{O(z\log z)}\cdot \tilde{O}(k\epsilon^{-z})$,
and $\forall G\in \mathcal{G}$, $\cost_z(G,c)\leq (\frac{\epsilon}{6z})^z\cdot \frac{\cost_z(P,c)}{k\cdot \log (24z/\epsilon)}$.
\end{enumerate}
\end{theorem}

Rings and groups are inherently different geometric objects,
hence they require different coreset construction methods. As in \cite{braverman2022power}, uniform sampling is applied on rings,
but a \emph{two-point coreset}, whose construction is defined in \Cref{def:twopoints}, is applied for each group.
Our main algorithm (\Cref{alg:main}) also follows this general strategy.

\begin{definition}[Construction of two-point coreset~\cite{braverman2022power}] \label{def:twopoints}
	For a group $G\subset \mathbb{R}^d$ and a center point $c\in \mathbb{R}^d$, let $\pfar^G$ and $\pclose^G$ denote the furthest and closest point to $c$ in $G$. For every $p\in G$, compute the unique $\lambda_p\in [0,1]$ such that $\dist^z(p,c)=\lambda_p\cdot \dist^z(\pclose^G,c)+(1-\lambda_p)\cdot \dist^z(\pfar^G,c)$. Let $D_G=\{\pfar^G,\pclose^G\}$, $w_{D_G}(\pclose^G)=\sum_{p\in G} \lambda_p$, and $w_{D_G}(\pfar^G)=\sum_{p\in G} (1-\lambda_p)$. $D_G$ is called the \emph{two-point} coreset of $G$ with respect to $c$.
\end{definition}

By definition, we can verify that $w_{D_G}(D_G) = |G|$ and $\cost_z(D_G, c) = \cost_z(G,c)$, which are useful for upper bounding the error induced by such two-point coresets.

\begin{figure}[t]
    \centering
    \includegraphics[height=0.2\textheight]{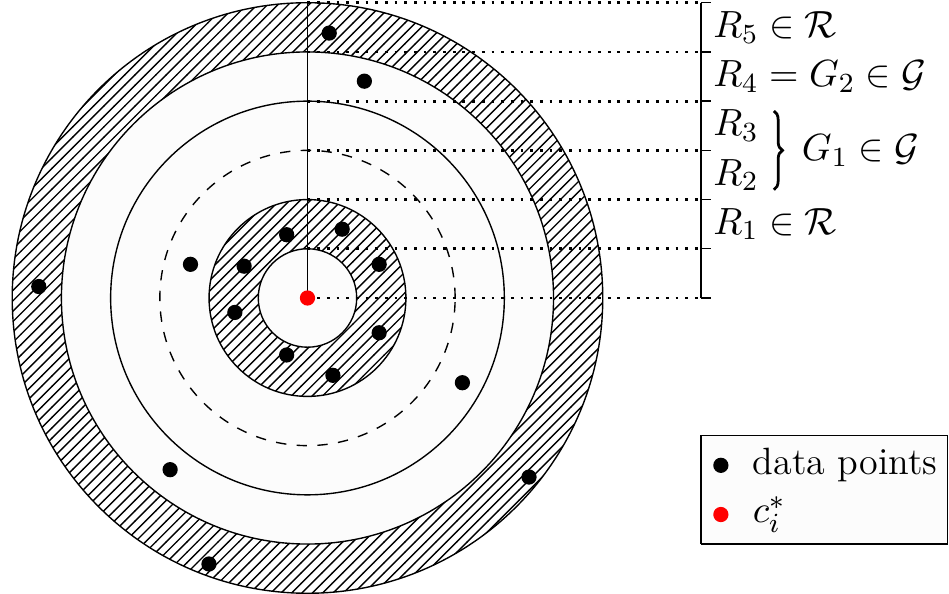}
    \caption{Illustration of \Cref{thm:meta} (plotted distance is the logarithm of the real distance).}
    \label{fig:ring_group}
\end{figure}

\subsection{Proof of Theorem~\ref{thm:main}}
\label{sec:proof_main}

\paragraph{Coreset Construction Algorithm}
We present our main algorithm in \Cref{alg:main}.
In Line~\ref{alg:L} and Line~\ref{alg:outlier}, the set $L^*$ of outliers of $C^*$ is the set of $\gamma m$ furthest points to $C^*$ and $L^*$ is directly added into the coreset $S$.
In Line~\ref{alg:cluster} and Line~\ref{alg:decomp}, the inliers $P\setminus L^*$ are decomposed into $\beta k$ clusters with respect to $C^*$ and the linear time decomposition algorithm of Theorem~\ref{thm:meta} is applied in each cluster.
In Line~\ref{alg:uniform} and Line~\ref{alg:twopoints}, similar to \cite{braverman2022power}, a uniform sampling and a two-point coreset (see \Cref{def:twopoints})
are applied in constructing coresets for rings and groups, respectively.

\begin{algorithm}
    \caption{Coreset Construction for \kzmC}
   \label{alg:main}
    \begin{algorithmic}[1]
        \Require dataset $P\subset \mathbb{R}^d$, $z \geq 1$, integer $k,m\geq 1$, an $(\alpha,\beta,\gamma)$-approximation $C^*=\{c_i^*\}_{i=1}^{\beta k}$
\State  let $L^*\gets \argmin_{|L|=\gamma m} \cost_z(P\setminus L,C^*)$ denote the set of $\gamma m$ outliers \label{alg:L}
        \State  add $L^*$ into $S$ and set $\forall x\in L^*, w_S \gets 1$ \label{alg:outlier}
        \State  partition $P\setminus L^*$ into $\beta k$ clusters $P_1,...,P_{\beta k}$ such that $P_i $ is the subset of $P \setminus L^*$ closest to $c_i^*$ \label{alg:cluster}
        \State for each $i\in [\beta k]$,  apply the decomposition of \Cref{thm:meta} to $(P_i,c_i^*)$ and obtain a collection $\mathcal{R}_i$ of disjoint rings and a collection $\mathcal{G}_i$ of disjoint groups \label{alg:decomp}
        \State for $i\in [\beta k]$ and every ring $R\in\mathcal{R}_i$, take a uniform sample $Q_{R}$ of size $2^{O(z\log z)}\cdot\tilde{O}(\frac{k}{\epsilon^{2z+2}})$ from $R$,
		set $\forall x\in Q_{R},w_{Q_{R}}(x) \gets \frac{|R|}{|Q_{R}|}$,
		and add $(Q_{R},w_{Q_{R}})$ into $S$\label{alg:uniform}
        \State  for $i\in [\beta k]$ and every group $G\in \mathcal{G}_i$ center $c_i^*$,
		construct a two-point coreset $(D_{G},w_{D_{G}})$ of $G$ as in \Cref{def:twopoints} and add $(D_{G},w_{D_{G}})$ into $S$
		\label{alg:twopoints}
        \State return $(S,w_S)$
    \end{algorithmic}
\end{algorithm}

\paragraph{Error Analysis}

Recall that $P$ is decomposed into 3 parts, the outliers $L^*$, the collection of rings, and the collection of groups.
We prove the coreset property for each of the 3 parts and claim the union yields an $\epsilon$-coreset of $P$ for \kzmC.
As $L^*$ is identical in the data set $P$ and the coreset $S$, we only have to put effort in the rings and groups.
We first introduce the following relaxed coreset definition which allows additive error.
\begin{definition}\label{def:additivecoreset}
Let $P\subset \mathbb{R}^d$, $0<\epsilon<1$ and $A\geq 0$, a weighted set $(S,w_S)$ is an $(\epsilon,A)$-coreset of $X$ for \kztC if for every $C\subset \mathbb{R}^d,|C|=k$,
$$
|\cost_z^{(t)}(P,C)-\cost_z^{(t)}(S,C)|\leq \epsilon\cdot \cost_z^{(t)}(P,C)+\epsilon \cdot A.
$$
\end{definition}
This allowance of additive error turns out to be crucial in our analysis,
and eventually we are able to charge the total additive error to the (near-)optimal cost, which enables us to obtain the coreset (without additive error).

The following lemma is a simple but useful way to bound the error between coresets and data sets and the proof idea is similar to Lemma 3.5 of \cite{braverman2022power} which relies on coupling the mass and applying generalized triangle inequality \Cref{lem:gentri}.
This technical lemma is used in many places in our entire proof.
\begin{lemma}
    \label{lem:naive}
	Let $B\subset P_i$ be either a ring or a group. Assume $(U,w_U)$ and $(V,w_V)$ are two weighted subsets of $B$ such that $w_U(U)=w_V(V)=N$, then for every $C\subset \mathbb{R}^d,|C|=k$ we have
	\begin{eqnarray}
	|\cost_z(U,C)-\cost_z(V,C)|\leq \epsilon \cdot \cost_z(U,C)+(\frac{6z}{\epsilon})^{z-1}\cdot \big(\cost_z(U,c_i^*)+\cost_z(V,c_i^*)\big).
	\end{eqnarray}
\end{lemma}

\begin{proof}
	Since $w_U(U)=w_V(V)$, there must exist a matching $M:U\times V\rightarrow \mathbb{R}_{\geq 0}$ between the mass of $U$ and $V$. So $\forall u\in U, \sum_{v\in V} M(u,v)=w_U(u)$ and $\forall v\in V, \sum_{u\in U} M(u,v)=w_V(v)$. By generalized triangle inequality \Cref{lem:gentri} we have,
	\begin{eqnarray*}
		&\quad&|\cost_z(U,C)-\cost_z(V,C)|\\
		&\leq& \sum_{u\in U}\sum_{v\in V}M(u,v)|\dist(u,C)^z-\dist(v,C)^z|\\
		&\leq &\sum_{u\in U}\sum_{v\in V}M(u,v)\big(\epsilon\cdot \dist(u,C)^z+(\frac{3z}{\epsilon})^{z-1}\cdot (\dist(u,C)-\dist(v,C))^z\big)\\
		&\leq &\epsilon\sum_{u\in U} w_U(u)\cdot \dist(u,C)^z+(\frac{3z}{\epsilon})^{z-1}\cdot\sum_{u\in U}\sum_{v\in V} M(u,v)\cdot (\dist(u,c_i^*)+\dist(v,c_i^*))^z\\
		&\leq &\epsilon \cdot \cost_z(U,C)+(\frac{3z}{\epsilon})^{z-1}\cdot (\sum_{u\in U} w_U(u)\cdot 2^{z-1}\cdot\dist(u,c_i^*)^z+\sum_{v\in V}w_V(v)\cdot 2^{z-1}\cdot \dist(v,c_i^*))\\
		&\leq &\epsilon\cdot \cost_z(U,C)+(\frac{6z}{\epsilon})^{z-1}\big(\cost_z(U,c_i^*)+\cost_z(V,c_i^*)\big)
	\end{eqnarray*}
\end{proof}

The following two are the key lemmas for the proof of \Cref{thm:main},
where we analyze the guarantee of
the uniform-sampling coresets for rings (\Cref{lem:intro:uniform}) and the two-point coresets (\Cref{lem:intro:twopoints}).

\begin{lemma}[Coresets for rings] \label{lem:intro:uniform}
Let $Q=\bigcup_{i\in[\beta k]}\bigcup_{R\in\mathcal{R}_i} Q_{R}$ denote the coreset of the rings $\Rall=\bigcup_{i\in [\beta k]}\bigcup_{R\in \mathcal{R}_i} R$, constructed by uniform sampling as in Line~\ref{alg:uniform} of Algorithm~\ref{alg:main}, then $\forall t,0\leq t \leq m$, $Q$ is an $\big(\epsilon,\cost_z(\Rall,C^*)\big)$-coreset of $\Rall$ for \kztC.
\end{lemma}
\begin{proof}
	The proof can be found in \Cref{sec:uniform}.
\end{proof}

\begin{lemma}[Two-point coresets for groups]\label{lem:intro:twopoints}
Let $D=\bigcup_{i\in[\beta k]}\bigcup_{G\in\mathcal{G}_i} D_{G}$ denote the two-point coresets of the groups $\Gall=\bigcup_{i\in [\beta k]}\bigcup_{G\in\mathcal{G}_i} G$, as in Line~\ref{alg:twopoints} of Algorithm~\ref{alg:main}, then for every $t,0\leq t\leq m$, $D$ is an $(\epsilon,\cost_z\big(P\setminus L^*,C^*)\big)$-coreset of $\Gall$ for \kztC.
\end{lemma}
\begin{proof}
	The proof can be found in \Cref{sec:twopoints}.
\end{proof}

\begin{proof}[Proof of Theorem~\ref{thm:main}]
Fix a center $C\subset \mathbb{R}^d,|C|=k$ and fix a $t\in [0,m]$, we first prove that
\[
\cost_z^{(t)}(S,C)\leq (1+\epsilon)\cost_z^{(t)}(P,C)+\epsilon\cdot \cost_z(P\setminus L^*,C^*).
\]
To this end, assume $L\subset X$ is the set of outliers for $C$. Let $t_R=|L\cap \Rall|,t_G=|L\cap \Gall|$. By \Cref{lem:intro:uniform} and \Cref{lem:intro:twopoints}, there exists weighted subset $T_Q\subset Q,T_D\subset D$ such that, $w_{T_Q}(T_Q)=t_R$, $w_{T_D}(T_D)=t_G$,
\begin{eqnarray}
\cost_z(Q-T_Q,C)\leq (1+\epsilon)\cost_z(\Rall-(L\cap \Rall),C)+\epsilon\cdot \cost_z(\Rall,C^*)
\end{eqnarray} and
\begin{eqnarray}
\cost_z(D-T_D,C)\leq (1+\epsilon)\cost_z(\Gall-(L\cap \Gall),C)+\epsilon\cdot \cost_z(P\setminus L^*,C^*)
\end{eqnarray}
Define a weighted subset $(T,w_T)$ of $S$, such that $T=(L\cap L^*)\cup T_Q\cup T_D$. Then $w_T(T)=t$ and
\begin{eqnarray*}
\cost_z^{(t)}(S,C)&\leq& \cost_z(S-T,C)\\
&=&\cost_z(L^*-(L\cap L^*),C)+\cost_z(Q-T_Q,C)+\cost_z(D-T_D,C)\\
&\leq& \cost_z(L^*-(L\cap L^*),C)+(1+\epsilon)\cost_z(\Rall-(L\cap \Rall),C)\\
&+&\epsilon\cdot \cost_z(\Rall,C^*)
+ (1+\epsilon)\cost_z(\Gall-(L\cap \Gall),C)\\
&+&\epsilon\cdot \cost_z(P\setminus L^*,C^*)\\
&\leq & (1+\epsilon)\cost_z(P-L,C)+O(\epsilon)\cdot \cost_z(P\setminus L^*,C^*)\\
&\leq & (1+O(\alpha\cdot \epsilon))\cost_z^{(t)}(P,C).
\end{eqnarray*}
Similarly, we can also obtain that
$\cost_z^{(t)}(P,C)\leq (1+O(\alpha\cdot \epsilon))\cost_z^{(t)}(S,C)$ for any $0\leq t\leq m$.
It remains to scale $\epsilon$ by a universal constant.
\end{proof}

\subsection{Proof of Lemma~\ref{lem:intro:uniform}: Error Analysis of Uniform Sampling} \label{sec:uniform}

As with recent works in Euclidean coresets \cite{cohen-addad2021improved,cohen2021new,cohen2022towards, braverman2022power}, we make use of an iterative size reduction \cite{BJKW21} and a terminal embedding technique \cite{NN19}, which allows us to trade the factor $O(d)$ in coreset size bound with a factor of $O\big(\frac{\log (k/\epsilon)}{\epsilon^2}\big)$.
Hence, it suffices to prove that a uniform sample of size $\tilde{O}(\frac{kd}{\epsilon^{2z}})$ yields the desired coreset.

The following simple formula can be obtained via integration by parts.
\begin{fact} \label{fact:integral}
Let $(Y,w_Y)$ denote a weighted dataset and $C\subseteq \mathbb{R}^d,|C|=k$ then for every $0\leq t\leq w_Y(Y)$ $$\cost^{(t)}_z(P,C)=\int_{0}^{\infty} z\cdot u^{z-1}\cdot \bigg(w_Y(Y)-m-w_Y\big(\mathrm{Balls}(C,u)\cap Y\big)\bigg)^+ du. $$
\end{fact}

The following notion of $\epsilon$-approximation for $k$-balls range space is well-studied in PAC learning and computational geometry communities (see e.g. \cite{GAA}).

\begin{definition}[$\epsilon$-Approximation for $k$-balls range space]
 Let $\Balls_k=\{\Balls(C,u)\mid C\subset \mathbb{R}^d,|C|=k,u>0\}$ denote the set of unions of $k$ balls with the same radius. For a dataset $P\subset \mathbb{R}^d$, the $k$-Balls range space on $P$ is denoted by $(P,\mathcal{P}_k)$ where $\mathcal{P}_k:=\{P\cap \Balls_k(C,u)\mid \Balls_k(C,u)\in \Balls_k\}$. A subset $Y\subset P$ is called an $\epsilon$-approximation of the $k$-Balls range space $(P,\mathcal{P}_k)$ if for every $\Balls(C,u)\in \Balls_k$,
 $$
 \bigg{|}\frac{|P\cap \Balls(C,u)|}{|P|}-\frac{|Y\cap \Balls(C,u)|}{|Y|}\bigg{|}\leq \epsilon.
 $$
\end{definition}
The following lemma reduces the construction of an $\epsilon$-approximation to uniform sampling.
\begin{lemma}[\cite{LI2001516}]
\label{lem:approximation}
Assume $Q$ is a uniform sample of size $\tilde{O}(\frac{kd}{\epsilon^2})$ from $P$, then with probability at least $1-\frac{1}{\poly(k/\epsilon)}$, $Q$ is an $\epsilon$-approximation of the $k$-Balls range space on $P$.
\end{lemma}

The following Lemma~\ref{lem:uniform} shows an $(\frac{\epsilon}{12z})^{z}$-approximation yields a $2^{O(z\log z)}\cdot \epsilon$-coreset for robust \kztC for every $t$.

\begin{lemma} \label{lem:uniform}
Assume $R=P_i\cap \ring(c_i^*,r,2r)$ is a ring in the cluster $P_i$. Let $Q_R$ be an $(\frac{\epsilon}{12z})^z$-approximation of the $k$-balls range space on $R$. Suppose every element of $Q_R$ is re-weighted by $\frac{|R|}{|Q_R|}$ then for every $C\subset \mathbb{R}^d, |C|=k$ and every $t,0\leq t\leq \min\left\{|R|,m\right\}$,
\begin{eqnarray}\label{eqn:uniform}
|\cost^{(t)}_z(R,C)-\cost^{(t)}_z(Q_R,C)|\leq \epsilon\cdot \cost_z^{(t)}(R,C)+\epsilon r^z |R|.
\end{eqnarray}
\end{lemma}

\begin{proof}
Fix a $C\subset \mathbb{R}^d,|C|=k$. As $Q_R$ is an $(\frac{\epsilon}{12z})^{z}$-approximation of the $k$-Balls range space on $R$, we know that for every $u>0$,
\begin{eqnarray} \label{eqn:uniform1}
|w_R\big(\mathrm{Balls}(C,u)\cap R)\big)-w_{Q_R}\big(\mathrm{Balls}(C,u)\cap Q_R\big)|\leq (\frac{\epsilon}{12z})^{z}\cdot  |R|.
\end{eqnarray}

 Let $\Tclose=\min_{x\in R} \dist(x,C)$ and $\Tfar=\max_{x\in R}\dist(x,C)$. Since $R\subset \ring(c_i^*,r,2r)$, the diameter of $R$ is at most $4r$ and this implies $\Tfar-\Tclose\leq 4r$. Since $Q_R$ is a subset of $R$, we know that for every $u\not\in [\Tclose,\Tfar]$,
 \begin{eqnarray} \label{eqn:uniform2}
 w_R\big(\mathrm{Balls}(C,u)\cap R)\big)=w_{Q_R}\big(\mathrm{Balls}(C,u)\cap Q_R\big)
 \end{eqnarray}

To prove (\ref{eqn:uniform}), we do the following case analysis.

If the number of outliers $t\geq (1-(\frac{\epsilon}{12z})^z)\cdot |R|$, let $L_R\subset R$ and $L_Q\subset Q_R$ denote the outliers of $R$ and $Q$ with respect to $C$. Using \Cref{lem:naive}, we know that
\begin{eqnarray*}
&\quad&|\cost^{(t)}_z(R,C)-\cost^{(t)}_z(Q_R,C)|\\
&\leq& \epsilon \cdot \cost_z^{(t)}(R,C)+(\frac{6z}{\epsilon})^{z-1}\cdot \big(\cost_z(R-L_R,c_i^*)+\cost_z(Q_R-L_Q,c_i^*)\big)\\
&\leq &\epsilon \cdot \cost_z^{(t)}(R,C)+(\frac{6z}{\epsilon})^{z-1} \cdot(\frac{\epsilon}{12z})^z\cdot |R|\cdot (2r)^z\\
&\leq& \epsilon \cdot \cost_z^{(t)}(R,C)+ \epsilon r^z|R|.
\end{eqnarray*}

If $t<(1-(\frac{\epsilon}{12z})^z)\cdot |R|$, using (\ref{eqn:uniform1}), (\ref{eqn:uniform2}), \Cref{fact:integral} and the generalized triangle inequality \Cref{lem:gentri}, we have,
 \begin{eqnarray*}
 &\quad&|\cost^{(t)}_z(P,C)-\cost^{(t)}_z(Q_R,C)|\\
 &\leq& \int_{\Tclose}^{\Tfar} zu^{z-1}\cdot \big|w_R\big(\mathrm{Balls}(C,u)\cap R)\big)-w_{Q_R}\big(\mathrm{Balls}(C,u)\cap Q_R\big)\big| du\\
 &\leq& (\frac{\epsilon}{12z})^z \cdot |R|\cdot (\Tfarz-\Tclosez)\\
 &\leq & (\frac{\epsilon}{12z})^z \cdot |R| \cdot \big(\epsilon\cdot \Tclosez+(\frac{3z}{\epsilon})^{z-1}\cdot (4r)^z\big)\\
 &\leq & \epsilon\cdot (\frac{\epsilon}{12z})^{z}\cdot  |R|\cdot\Tclosez+ \epsilon r^z|R|\\
 &\leq& \epsilon \cdot \cost_z^{(t)}(R,C)+\epsilon r^z|R|
 \end{eqnarray*}
 where for the last inequality, we have used the fact that
 \[
        \cost_z^{(t)}(R,C)\geq (|R|-t)\cdot \Tclosez\geq (\frac{\epsilon}{12z})^z \cdot |R|\cdot \Tclosez.
\]
\end{proof}

We are ready to prove \Cref{lem:intro:uniform}.
\begin{proof}[Proof of \Cref{lem:intro:uniform}] Fix a center $C\subset \mathbb{R}^d,|C|=k$.
By \Cref{lem:approximation}, the sample size in Line~\ref{alg:uniform} of \Cref{alg:main} implies that $Q_R$ is an $(\frac{\epsilon}{12z})^{z}$-approximation of the $k$-balls range space on $R$ for every $R\in\bigcup_{i\in[ \beta k]} \mathcal{R}_i$ . By \cref{lem:uniform} and the union bound, with probability at least $0.9$, for every $i\in [\beta]$, for every ring $R\in \mathcal{R}_i$, and for every $e\in [0,|R|]$,
\begin{eqnarray} \label{eqn:uniform:QR}
|\cost^{(e)}_z(R,C)-\cost^{(e)}_z(Q_R,C)|\leq \epsilon\cdot \cost_z^{(e)}(R,C)+\epsilon\cdot \cost_z(R,c_i^*).
\end{eqnarray}

Let $L$ denote the set of $t$ outliers of $\Rall$ with respect to $C$. By (\ref{eqn:uniform:QR}), for every $R\in \bigcup_{i\in [\beta k]} \mathcal{R}_i$, there exists a weighted subset $T_R\subset Q_R$ such that $w_{T_R}(T_R)=w_L(L\cap R)$ and
$$
\cost_z(Q_R-T_R,C)\leq (1+\epsilon)\cdot \cost_z(R-(L\cap R),C)+\epsilon \cdot \cost_z(R,c_i^*).
$$

Summing over all $R\in\bigcup_{i\in [\beta k]} \mathcal{R}_i$, we know that,
\begin{eqnarray*}
&\quad&\cost_z^{(t)}(Q,C)\\
&\leq& \sum_{i\in [\beta k]}\sum_{R\in \mathcal{R}_i} \cost_z(Q_R-T_R,C)\\
&\leq& \sum_{i\in [\beta k]}\sum_{R\in \mathcal{R}_i} \big((1+\epsilon)\cdot \cost_z(R-(L\cap R),C)+\epsilon\cdot \cost_z(R,c_i^*)\big)\\
&=&(1+\epsilon)\cdot \cost_z(\Rall-L,C)+\epsilon\cdot \cost_z(\Rall,C^*)\\
&=&(1+\epsilon)\cdot \cost_z^{(t)}(\Rall,C)+\epsilon\cdot \cost_z(\Rall,C^*).
\end{eqnarray*}

On the same way, we can show that
$$
\cost_z^{(t)}(\Rall,C)\leq (1+\epsilon)\cdot \cost_z^{(t)}(Q,C)+\epsilon\cdot \cost_z(\Rall,C^*).
$$

Thus we finish the proof.
\end{proof}  
\subsection{Proof of Lemma~\ref{lem:intro:twopoints}: Error Analysis of Two-Point Coresets} \label{sec:twopoints}

\begin{figure}[t]
    \centering
    \includegraphics[height=0.2\textheight]{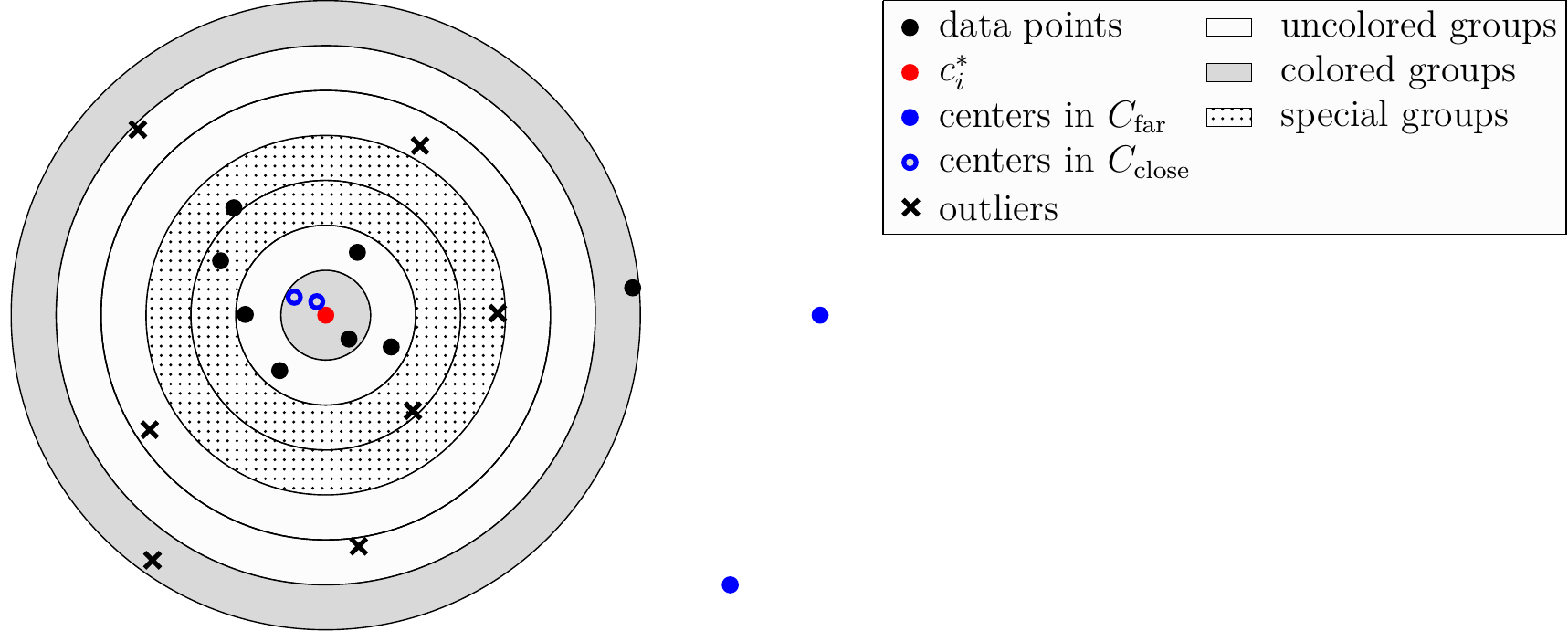}
    \caption{An illustration of the decomposition into colored, uncolored and special groups with respect to $C=C_\mathrm{far}\cup C_\mathrm{close}$, where the radii of balls are taken the logarithm.}
    \label{fig:colored_uncolered}
\end{figure}

Throughout this section, we fix a center set $C\subset \mathbb{R}^d,|C|=k$ and prove the coreset property of $D$ with respect to $C$.
To analyze the error of two-point coreset for $\Gall$, we further decomposes all groups into colored groups and uncolored groups based on the position of $C$ in the following \Cref{lem:color}, which was also considered in~\cite{braverman2022power}.
Furthermore, inside our proof, we also consider a more refined type of groups called \emph{special groups}. An overview illustration of these groups and other relevant notions can be found in \Cref{fig:colored_uncolered}.

\begin{lemma}[Colored groups and uncolored groups~\cite{braverman2022power}]
\label{lem:color}
For a center set $C\subset\mathbb{R}^d,|C|=k$, a collection of groups $\mathcal{G}_i$ can be further divided into \emph{colored} groups and \emph{uncolored} groups with respect to $C$ such that
\begin{enumerate}
\item there are at most $O(k\log\frac{z}{\epsilon})$ colored groups and
\item for every uncolored group $G\in \mathcal{G}_i$, for every $u\in C$, either $\forall p\in G,\dist(u,c_i^*)<\frac{\epsilon}{9z}\cdot \dist(p,c_i^*)$ or $\forall p\in G, \dist(u,c_i^*)>\frac{24z}{\epsilon}\dist(p,c_i^*)$.
\end{enumerate}
\end{lemma}

Let $G\in \mathcal{G}_i$ be an uncolored group with respect to $C$, \Cref{lem:color} implies that the center set $C$ can be decomposed into a ``close'' portion and a ``far'' portion to $G$, as in the following \Cref{def:Cfar}.

\begin{definition}[\cite{braverman2022power}] \label{def:Cfar}
For a center set $C$, assume $G\in \mathcal{G}_i$ is an uncolored group with respect to $C$. Define $$
\Cfar^{G}=\{u\in C\mid \forall p\in G,\dist(u,c_i^*)>\frac{24z}{\epsilon}\dist(p,c_i^*)\},$$ and
$$
\Cclose^{G}=\{u\in C\mid \forall p\in G,\dist(u,c_i^*)<\frac{\epsilon}{9z}\cdot\dist(p,c_i^*)\}.
$$
Remark that $C=\Cfar^{G}\cup \Cclose^{G}$ by the property of uncolored group.
\end{definition}

The following \Cref{lem:additive} shows the difference of cost to any center $C\subset\mathbb{R}^d,|C|=k$ between a group $G$ and its two-point coreset $D_{G}$ can always be bounded by a small additive error, via generalized triangle inequality \Cref{lem:gentri}.

By combining \Cref{lem:naive} and the fact that $\cost_z(G_i,c_i^*)\leq (\frac{\epsilon}{6z})^z\cdot \frac{\cost_z(P_i,c_i^*)}{k\log (24z/\epsilon)}$, we can obtain the following inequality.

\begin{lemma}[Robust variant of~{\cite[Lemma 3.5]{braverman2022power}}]\label{lem:additive}
For a group $G\in \mathcal{G}_i$, assume $(U,w_U)$ and $(V,w_V)$ are two weighted subsets of $G$ such that $w_U(U)=w_V(V)$. Then for every $C\subset\mathbb{R}^d,|C|=k$,
\begin{eqnarray}
|\cost_z(U,C)-\cost_z(V,C)|\leq \epsilon \cdot \cost_z(U,C)+\epsilon\cdot \frac{ \cost_z(P_i,c^*)}{2k\log(z/\epsilon)}.
\end{eqnarray}
\end{lemma}

\begin{lemma} \label{lem:uncol}
Let $G$ denote an uncolored group with respect to $C$. Suppose $(U,w_U)$ and $(V,w_V)$ are two weighted subsets of $G$ such that one of the following items hold,
\begin{enumerate}
\item either $\Cclose^G\not=\emptyset$ and $\cost_z(U,c_i^*)=\cost_z(V,c_i^*)$,
\item or $\Cclose^{G}=\emptyset$ and $w_U(U)=w_V(V)$.
\end{enumerate}

Then we have
\begin{eqnarray} \label{eqn:uncol:mul}
\cost_z(U,C)\in (1\pm \epsilon)\cost_z(V,C).
\end{eqnarray}
\end{lemma}
\begin{proof}
If $\Cclose^{G}\not=\emptyset$, by the property of uncolored group as in \Cref{lem:color}, we know that $\forall x\in G,\dist(x,C)\in (1\pm\frac{\epsilon}{3z})\cdot \dist(x,c_i^*)$. So we have
$$
\cost_z(U,C)\in (1\pm\epsilon)\cost_z(U,c_i^*)\;\;\mathrm{and}\;\;\cost_z(V,C)\in(1\pm\epsilon)\cost_z(V,c_i^*).
$$ By combining the above two inequalities and scaling $\epsilon$, we obtain (\ref{eqn:uncol:mul}).

In the other case, if $\Cclose^{G}=\emptyset$, \Cref{lem:color} implies $\forall x\in G, \dist(x,C)>\frac{9z}{\epsilon}\cdot \dist(x,c_i^*)$. By triangle inequality, we know that $\dist(x,C)\in (1\pm \frac{\epsilon}{3z})\dist(c_i^*,C)$. So we have,
$$
\cost_z(U,C)\in (1\pm\epsilon)\cdot w_U(U)\cdot \cost_z(c_i^*,C)\;\;\mathrm{and}\;\;\cost_z(V,C)\in(1\pm\epsilon)\cdot w_V(V)\cdot \cost_z(c_i^*,C),
$$
moreover since $w_U(U)=w_V(V)$, we conclude (\ref{eqn:uncol:mul}) by scaling $\epsilon$.
\end{proof}

We are ready to prove \Cref{lem:intro:twopoints}.

\begin{proof}[Proof of \Cref{lem:intro:twopoints}]
It suffices to prove the following two directions separately.
\begin{eqnarray} \label{eqn:tp1}
\cost_z^{(t)}(D,C)\leq (1+\epsilon)\cost_z^{(t)}(\Gall,C)+\epsilon \cdot \cost_z(P\setminus L^*,C^*),
\end{eqnarray}
\begin{eqnarray} \label{eqn:tp2}
\cost_z^{(t)}(\Gall,C)\leq (1+\epsilon)\cost_z^{(t)}(D,C)+\epsilon \cdot \cost_z(P\setminus L^*,C^*),
\end{eqnarray}
and scale $\epsilon$.
\end{proof}

\paragraph{Proof of (\ref{eqn:tp1})}

Let $(L,w_L)$ denote the outliers of $\Gall$ with respect to $C$. Namely, $L\subset G, w_L(L)=t$ and $$
\cost_z(\Gall-L,C)=\cost_z^{(t)}(\Gall,C).
$$

It suffices to find a weighted subset $(T,w_T)$ of $D$ such that $w_T(T)=t$ and
\begin{eqnarray}\label{eqn:d1}
\cost_z(D-T,C)\leq (1+\epsilon)\cost_z(\Gall-L,C)+\epsilon\cdot \cost_z(P\setminus L^*,C^*).
\end{eqnarray}

We define $T$ as the following. Recall that $$
\Gall=\bigcup_{i\in [\beta k]}\bigcup_{G\in\mathcal{G}_i} G.
$$
For every $G\in\mathcal{G}_i$, we add $\{\pclose^{G},\pfar^{G}\}$ into $T$ and set $$
w_{T}(\pclose^{G})=\sum_{x\in L\cap G} \lambda_x,\;\;\; w_{T}(\pfar^{G})=\sum_{x\in L\cap G} (1-\lambda_x)
$$ where we recall that $\lambda_x$ is the unique number in $[0,1]$ such that $\dist^z(x,c_i^*)=\lambda_x\cdot \dist^z(\pclose^G,c_i^*)+(1-\lambda_x)\cdot\dist^z(\pfar^G,c_i^*)$.

If $G$ is an colored group, we apply \Cref{lem:additive} to obtain
\begin{eqnarray} \label{eqn:d1:col}
\cost_z(D_G-(T\cap D_{G}),C)\leq (1+\epsilon)\cost_z(G-(L\cap G),C)+\epsilon\cdot \frac{\cost_z(P_i,c_i^*)}{2k\log (z/\epsilon)}
\end{eqnarray}

Now suppose $G$ is an uncolored group, observe that by construction, $w_T(T\cap D_{G})=w_L(L\cap G)$ and $\cost_z(T\cap D_{G},c_i^*)=\cost_z(L\cap G,c_i^*)$. Applying \Cref{lem:uncol} in $D_{G}-(T\cap D_{G})$ and $G-(L\cap G)$, we obtain that,
\begin{eqnarray} \label{eqn:d1:uncol}
\cost_z(D_{G}-(T\cap D_{G},C)\leq (1+\epsilon)\cost_z(G-(L\cap G),C).
\end{eqnarray}

By \Cref{lem:color}, there are at most $k\log (z/\epsilon)$ many colored groups in each cluster $P_i$, combining with (\ref{eqn:d1:col}) and (\ref{eqn:d1:uncol}), we have
\begin{eqnarray*}
&\quad&\cost_z(D-T,C)\\
&=&\sum_{i\in [\beta k]}\sum_{G\in\mathcal{G}_i} \cost_z(D_{G}-(T\cap D_{G}),C)\\
&\leq & \sum_{i\in [\beta k]}\sum_{G\in\mathcal{G}_i} (1+\epsilon) \cost_z(G-(L\cap G),C)+k\log(z/\epsilon)\sum_{i\in [\beta k]} \epsilon\cdot \frac{\cost_z(P_i,c_i^*)}{2k\log (z/\epsilon)}\\
&\leq & (1+\epsilon)\cost_z(\Gall-L,C)+\epsilon\cdot \cost_z(P\setminus L^*,C^*)
\end{eqnarray*}
which is (\ref{eqn:d1}).

\paragraph{Proof of (\ref{eqn:tp2})}

Let $(T,w_T)$ denote the set of (total weight $w_T(T)=t$) outliers of $D$ with respect to $C$. Namely, $$
\cost_z^{(t)}(D,C)=\cost_z^{(t)}(D-T,C).
$$

It suffices to find a weighted subset $(L,w_L)$ of $G$ such that $w_L(L)=t$ and
\begin{eqnarray}\label{eqn:d2}
\cost_z^{(t)}(\Gall-L,C)\leq (1+\epsilon)\cdot \cost_z^{(t)}(D-T,C)+\epsilon\cdot \cost_z(P\setminus L^*,C^*).
\end{eqnarray}

We construct $L$ as the following. For every $i\in [k]$, for every $G\in \mathcal{G}_i$, let $m_{G}=w_{T}(T\cap D_{G})$ and let $(L_{G},w_{(L_{G})})$ denote a weighted subset of $G$ such that $$
 \cost^{(m_{G})}_z(G,C)=\cost^{(m_{G})}_z(G-L_{G},C).
 $$ In other words, $L_{G}$ is the subset of furthest $m_{G}$ weights of points to $C$ in $G$.
 Add $L_{G}$ into $L$ and set $w_{L}(x)=w_{(L_{G})}(x)$ for every $x\in L_{G}$.

We prove $L$ satisfies (\ref{eqn:d2}). We do the following case study.

\begin{itemize}
\item If $G$ is a colored group, we simply apply \Cref{lem:additive} to obtain
\begin{eqnarray} \label{eqn:d2:c1}
\cost_z(G-L_{G},C)\leq (1+\epsilon)\cost_z(D_{G}-(T\cap D_{G}),C)+\epsilon\cdot \frac{\cost_z(P_i,c_i^*)}{2k\log (z/\epsilon)}.
\end{eqnarray}

\item If $G$ is an uncolored group, and $\Cclose^{G}=\emptyset$, by \Cref{lem:uncol}, we know that
\begin{eqnarray}  \label{eqn:d2:c2}
\cost_z(G-L_{G},C)\leq (1+\epsilon)\cost_z(D_{G}-(T\cap D_{G}),C).
\end{eqnarray}

\item If $G$ is an uncolored group, $\Cclose^{G}\not=\emptyset$, and $m_{G}\in \{0,|G|\}$, note that in this case $L_{G}=G$ or $L_{G}=\emptyset$. So we have
\begin{eqnarray}
\cost_z(G-L_{G},c_i^*)=\cost_z(D_{G}-(T\cap D_{G}),c_i^*)
\end{eqnarray}
by the fact that $D_{G}$ is the two-point coreset of $G$, satisfying \Cref{def:twopoints}.
So in this case, the conditions of \Cref{lem:uncol} are satisfied. So we have,
\begin{eqnarray} \label{eqn:d2:c3}
\cost_z(G-L_{G},C)\leq (1+\epsilon)\cost_z(D_{G}-(T\cap D_{G}),C).
\end{eqnarray}

\item If $G$ is an uncolored group, $\Cclose^{G}\not=\emptyset$, and $m_{G}\not\in\{0,|G|\}$, we call such group a \emph{special} uncolored group and prove in \Cref{lem:special} that there at most $2$ special groups in every $\mathcal{G}_i$. (See \Cref{fig:colored_uncolered} for an illustration.) Then we use \Cref{lem:additive} to obtain
    \begin{eqnarray} \label{eqn:d2:c4}
    \cost_z(G-L_{G},C)\leq (1+\epsilon)\cost_z(D_{G}-(T\cap D_{G}),C)+\epsilon\cdot \frac{\cost_z(P_i,c_i^*)}{2k\log (z/\epsilon)}.
    \end{eqnarray}

\end{itemize}

Combining (\ref{eqn:d2:c1}), (\ref{eqn:d2:c2}), (\ref{eqn:d2:c3}), (\ref{eqn:d2:c4}), and the fact that there are at most $k\log(z/\epsilon)$ colored groups and $2$ special groups in each $\mathcal{G}_i$, we have
\begin{eqnarray*}
&\quad&\cost_z(\Gall-L,C)\\
&=&\sum_{i\in[\beta k]}\sum_{G\in \mathcal{G}_i}\cost_z(G-L_{G},C)\\
&\leq& (1+\epsilon) \sum_{i\in[\beta k]}\sum_{G\in \mathcal{G}_i}\cost_z(D_{G}-(T\cap D_{G}),C)+(k\log(z/\epsilon)+2)\cdot\sum_{i\in [\beta k]} \epsilon\cdot \frac{\cost_z(P_i,c_i^*)}{2k\log (z/\epsilon)}\\
&\leq&(1+\epsilon)\cost_z(D-T,C)+\epsilon\cdot \cost_z(P\setminus L^*,C^*).
\end{eqnarray*}

\begin{lemma} \label{lem:special}
For a center set $C\subset \mathbb{R}^d$, $|C|=k$, in every $\mathcal{G}_i$, there are at most $2$ special uncolored groups with respect to $C$.
\end{lemma}

\begin{proof}

For the sake of contradiction, assume there are 3 special uncolored groups $G_{1},G_{2}$, and $G_{3}$ in cluster $P_i$. Assume w.l.o.g. that $G_{1}$ is the furthest to center $c_i^*$ and $G_{3}$ is the closest one.
Since $G_{1}$ is a special uncolored group, we know that $\Cclose^{G_1}\not=\emptyset$, so $\forall x\in G_{1}$, $\dist(x,C)\in \big(1\pm \epsilon\big)\cdot \dist(x,c_i^*)$. In particular, there exists an inlier $y_1\in D_{G_1}$ such that $\dist(y_1,C)\geq (1-\epsilon)\cdot \dist(y_1,c_i^*)$. Similarly, there exists an outlier $y_3\in G_{3}$ such that $\dist(y_3,C)\leq (1+\epsilon)\cdot \dist(y_3,c_i^*)$.

However, $G_{1}$, $G_{2}$ and $G_{3}$ are disjoint groups which are union of consecutive rings. So $\dist(y_1,c_i^*)\geq 2 \dist(y_3,c_i^*)$ and this implies
\begin{eqnarray*}
\dist(y_1,C)&\geq& (1-\epsilon)\cdot \dist(y_1,c_i^*)\\
&\geq& 2(1-\epsilon)\cdot \dist(y_3,c_i^*)\\
&>&(1+\epsilon)\cdot \dist(y_3,c_i^*)\\
&>&\dist(y_3,C)
\end{eqnarray*}
where we have used that $\epsilon<0.3$. However, this contradicts to the fact that $y_1$ is an inlier but $y_3$ is an outlier.
\end{proof}

 \section{Lower Bounds}
\label{sec:lb}

We show in \Cref{thm:lb} that the factor $m$ is necessary in the coreset size,
even for the very simple case of $k = 1$ and one dimension, for \kmMedian.

\begin{theorem}
	\label{thm:lb}
	For every integer $m \geq 1$, there exists a dataset $X \subset \mathbb{R}$ of $n \geq m$ points,
	such that for every $0< \epsilon < 0.5$,
	any $\eps$-coreset for \onemMedian must have size $\Omega(m)$.
\end{theorem}

\begin{proof}
	Fix $0< \eps < 0.5$.
Consider the following instance $X=\left\{x_0,\ldots, x_m\right\}\subset \R^1$ of size $n=m+1$, where $x_0 = 0$ and $x_i = i$ for $i\in [m]$.
Suppose $(S,w_S)$ is an $\eps$-coreset for \kmMedian.

	We first claim that $w_S(S)\geq m+1-\eps$. 
This can be verified by letting center $c\rightarrow +\infty$, and we have
	\[
	\frac{\cost_1^{(m)}(S,c)}{\cost_1^{(m)}(X,c)} = w_S(S)-m\in 1\pm \eps.
	\]
Next, let $c = \frac{x_{i-1}+x_i}{2}$ for some $i\in [m+1]$, which implies that $\cost_1^{(m)}(X,c) = |x_i - c| = 0.5$, i.e., the distance to the nearest-neighbor of $c$ in $X$. 
Suppose both $x_{i-1}$ and $x_i$ are not in $S$ and we have
	\begin{align*}
		\cost_1^{(m)}(S,c) &\geq && (w_S(S)-m)\cdot \min_{x\in S} d(x,c) & \\
		& \geq && (1-\eps)\cdot \min_{x\in S} d(x,c) & (w_S(S)\geq m+0.6)\\
		& \geq && (1-\eps)\cdot 1.5 & (x_{i-1},x_i\notin S) \\
		& > && 0.75 & (\eps < 0.5) \\
		& > && (1+\eps)\cdot \cost_1^{(m)}(X,c), & (\eps < 0.5)
	\end{align*}
	which is a contradiction.
Hence, either $x_{i-1}$ or $x_i$ must be contained in $S$.
It is not hard to conclude that $|S| \geq \frac{m-1}{2}$, which completes the proof.
\end{proof} \section{Experiments}
\label{sec:exp}

We implement our coreset construction algorithm and evaluate its empirical performance on various real datasets.
We compare it with several baselines and demonstrate the superior performance of our coreset.
In addition, we show that our coresets can significantly speed up approximation algorithms for both \kmMedian and \kmMeans problems.

\paragraph{Experiment Setup}
Our experiments are conducted on publicly available clustering datasets, see \Cref{tab:data} for a summary of specifications and choice of parameters.
For all datasets, we select numerical features to form a vector in $\R^d$ for each record. For larger dataset, particularly Census1990 and Twitter,
we subsample it to $10^5$ points so that inefficient baselines can still finish in a reasonable amount of time.
Unless otherwise specified, we typically set $k=5$ for the number of centers.
The number of outliers $m$ is determined by a per-dataset basis, 
via observing the distance distribution of points to a near-optimal center (see \Cref{sec:determine_outlier} for details).
All experiments are conducted on a PC with Intel Core i7 CPU and 16 GB memory, and algorithms are implemented using C++ 11.

\paragraph{Implementation Details}
Our coreset implementation mostly follows \Cref{alg:main} except for a few modifications.
For efficiency, we use a near-linear time algorithm by~\cite{bhaskara2019greedy} to compute an $(O(1), O(1), O(1))$-approximation (as required by \Cref{alg:main}),
but we still add $m$ outliers to coreset (in Line~\ref{alg:outlier}) instead of adding all the found ones.
Moreover, since it is more practical to directly set the target coreset size $N$
(instead of solving for $N$ from $\eps$), we modify the algorithm so that the generated coreset has exactly $N$ points.
Specifically, the coreset size is affected by two key parameters,
one is a threshold, denoted as $t$,
used to determine how the rings and groups are formed in the construction of \Cref{thm:meta} (whose details can be found in ~\cite{braverman2022power}),
and the other, denoted as $s$, is the size of each uniform sample (used in Line~\ref{alg:uniform}).
Here, we heuristically set $t = O(\frac{1}{N-m})$ and solve for $s$ such that the total size equals to $N$.

\begin{table}[t]
    \caption{Specifications of datasets and the choice of the parameters.}
    \label{tab:data}
    \centering
    \begin{tabular}{lllll}
        \toprule
        dataset & size & subsample & dim. & \# of outliers $m$ \\
        \midrule
        Adult ~\cite{adult_data}  & $48842$ & - & $6$ & $200$ \\
        Bank ~\cite{bank_data} & $41188$ & - &$10$ & $200$ \\
        Twitter~\cite{twitter_data} & $21040936$ & $10^5$ & $2$ & $500$ \\
        Census1990~\cite{census1990_data} & $2458285$ & $10^5$ & $68$ & $1500$ \\
        \bottomrule
    \end{tabular}
\end{table}

\subsection{Tradeoff between Size and Empirical Error}
\paragraph{Empirical Error}
We evaluates the tradeoff between the coreset size and empirical error
under the \kmMedian objective.
In general, for \kzmC,
given a coreset $S$, define its empirical error, denoted as $\hat{\eps}(S, C)$,
for a specific center $C \subset \mathbb{R}^d$, $|C| = k$ as
$
    \hat{\eps}(S, C) := \frac{|\cost_z^{(m)}(X, C) - \cost_z^{(m)}(S, C)|}{\cost_z^{(m)}(X, C)}.
$
Since it is difficult to exactly verify whether a coreset
preserves the objective for \emph{all} centers (as required by the definition),
we evaluate the empirical error, denoted as $\hat{\eps}(S)$, for the coreset $S$ as the maximum empirical error over $\mathcal{C}$,
which is a collection of $500$ randomly-chosen center sets,
i.e., $\hat{\eps}(S) := \max_{C \in \mathcal{C}} \hat{\eps}(S, C)$.
Note that $\hat{\eps}(S)$ is defined in a way similar to the worst-case error parameter $\eps$ as in \Cref{def:coreset}.

\paragraph{Baselines}
We compare our coreset with the following baselines: a) uniform sampling (US), where we draw $N$ independent uniform samples from $X$ and set the weight $\frac{|X|}{N}$ for each sample,
b) outlier-aware uniform sampling (OAUS),
where we follow Line~\ref{alg:L} - Line~\ref{alg:outlier} of \Cref{alg:main} to add $m$ outliers $L^*$ to the coreset and sample $N-m$ data points from $X\setminus L^*$ as in US baseline,
and c) sensitivity sampling (SS), the 
previous coreset construction algorithm of~\cite{feldman2012data}.

\paragraph{Experiment: Size-error Tradeoff}
In our experiment, for each coreset algorithm,
we run it to construct the coresets with varying target sizes $N$, ranging from $m + 300$ to $m + 4800$, with a step size of $500$.
We evaluate the empirical error $\hat{\eps}(\cdot)$
and we plot the size-error curves in \Cref{fig:size_vs_error} for each baseline and dataset.
To make the measurement stable, the coreset construction and evaluations are run $100$ times independently and the average is reported.
As can be seen from \Cref{fig:size_vs_error},
our coreset admits a similar error curve regardless of the dataset,
and it achieves about $2.5\%$ error using a coreset of size $m + 800$ (within $2.3\%$ - $2.5\%$ of data size),
which is perfectly justified by our theory that the coreset size only depends on $O(m + \poly(k\epsilon^{-1}))$.
Our coresets outperform all three baselines by a significant margin in every dataset and every target coreset size.
Interestingly, the two baselines SS and US seem to perform similarly,
even though the construction of SS~\cite{feldman2012data} is way more costly
since its running time has an exponential dependence on $k+m$, which is already impractical in our setting of parameters.
Another interesting finding is that, OAUS performs no better than US overall,
and both are much worse than ours.
This indicates that it is not the added initial outliers $L^*$ (as in \Cref{alg:main})
that leads to the superior performance of our coreset.
Finally, we also observe that our coreset has a smaller variance in the empirical error ($\approx 10^{-6}$), compared with other baselines ($\approx 10^{-4}$).

\begin{figure}[t]
    \centering
    \begin{subfigure}[b]{0.42\textwidth}
        \centering
        \includegraphics[width=\textwidth]{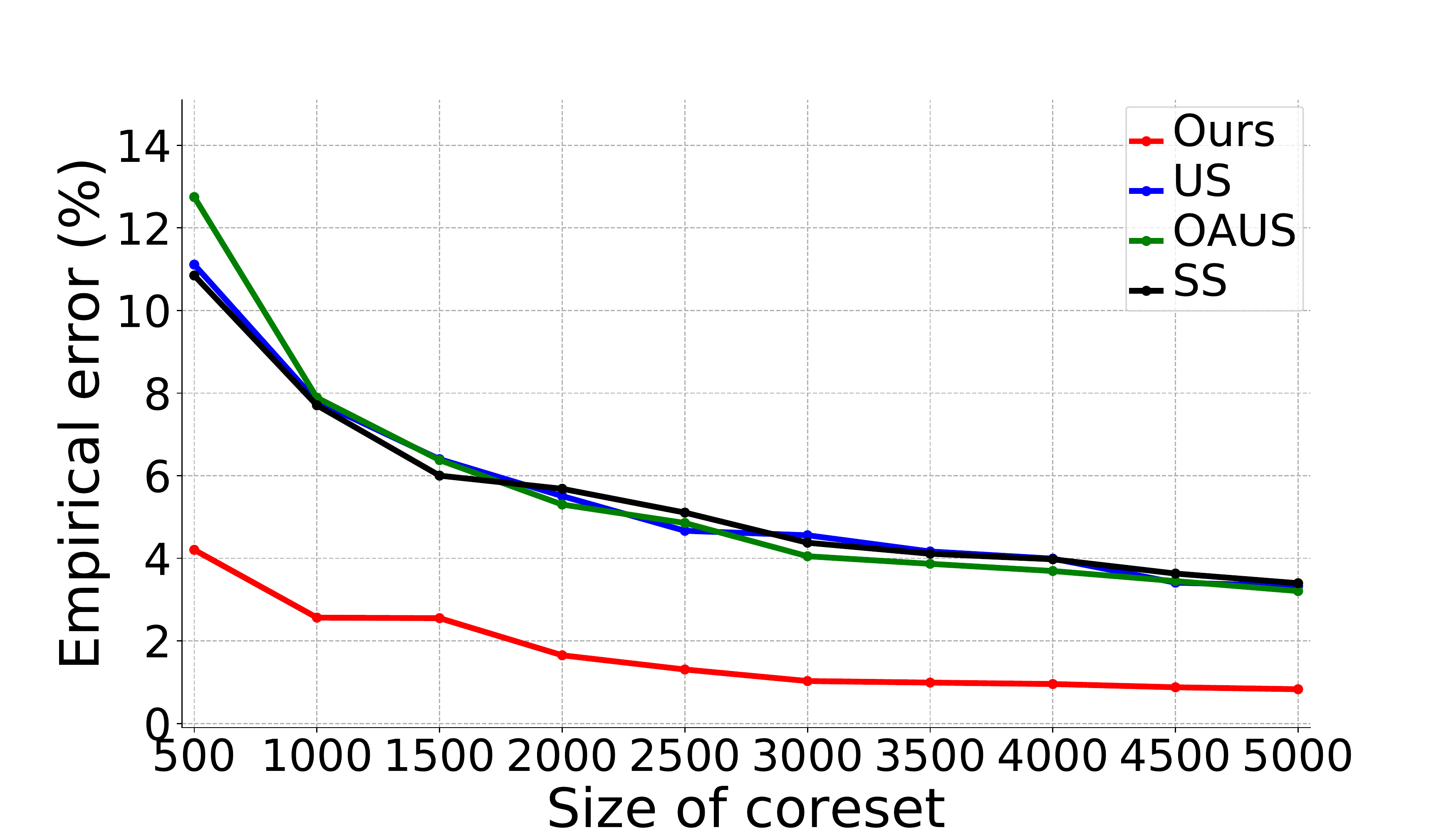}
        \caption*{Adult dataset}
    \end{subfigure} \qquad
    \begin{subfigure}[b]{0.42\textwidth}
        \centering
        \includegraphics[width=\textwidth]{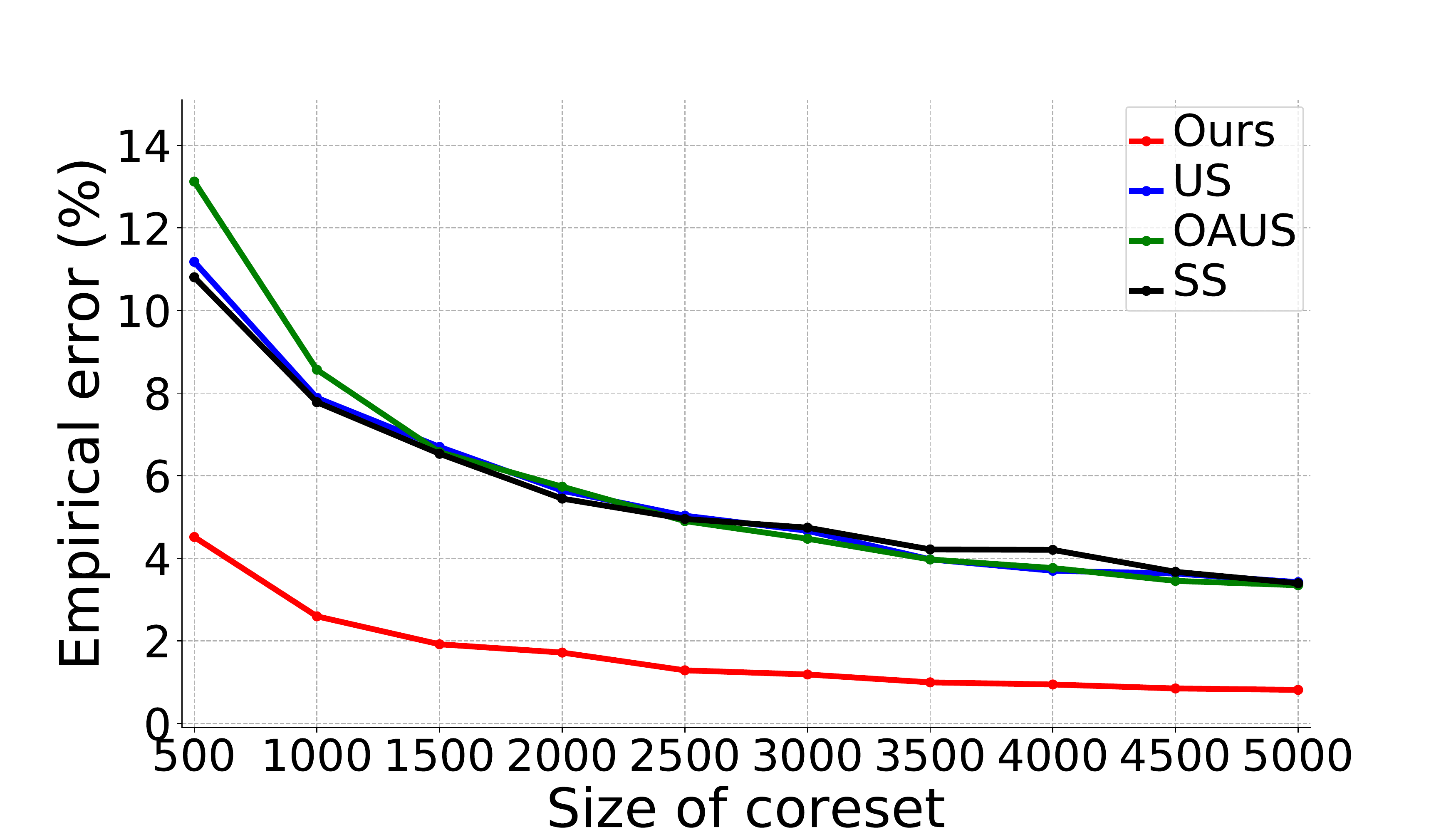}
        \caption*{Bank dataset}
    \end{subfigure}
    \begin{subfigure}[b]{0.42\textwidth}
        \centering
        \includegraphics[width=\textwidth]{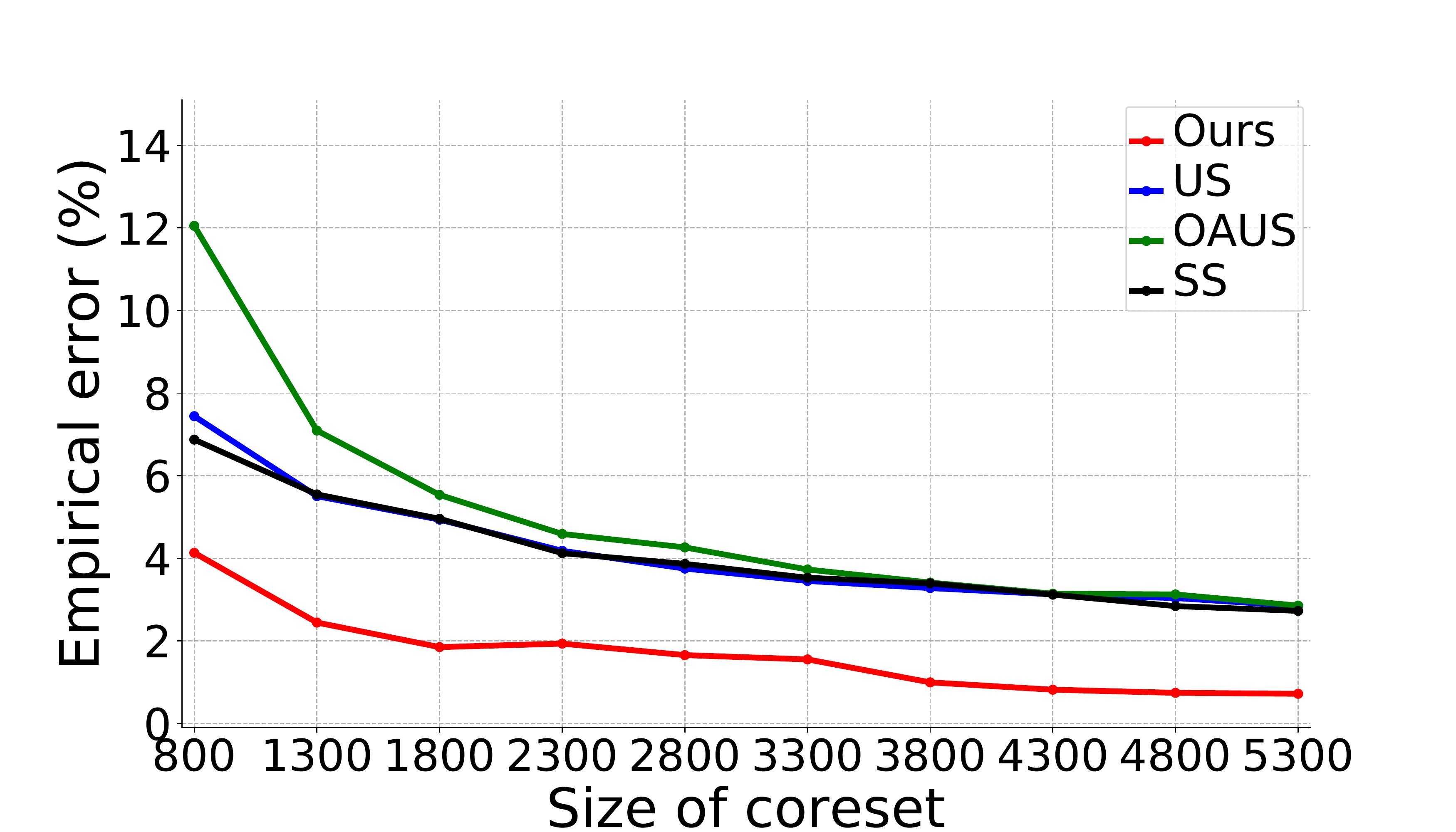}
        \caption*{Twitter dataset}
    \end{subfigure} \qquad
    \begin{subfigure}[b]{0.42\textwidth}
        \centering
        \includegraphics[width=\textwidth]{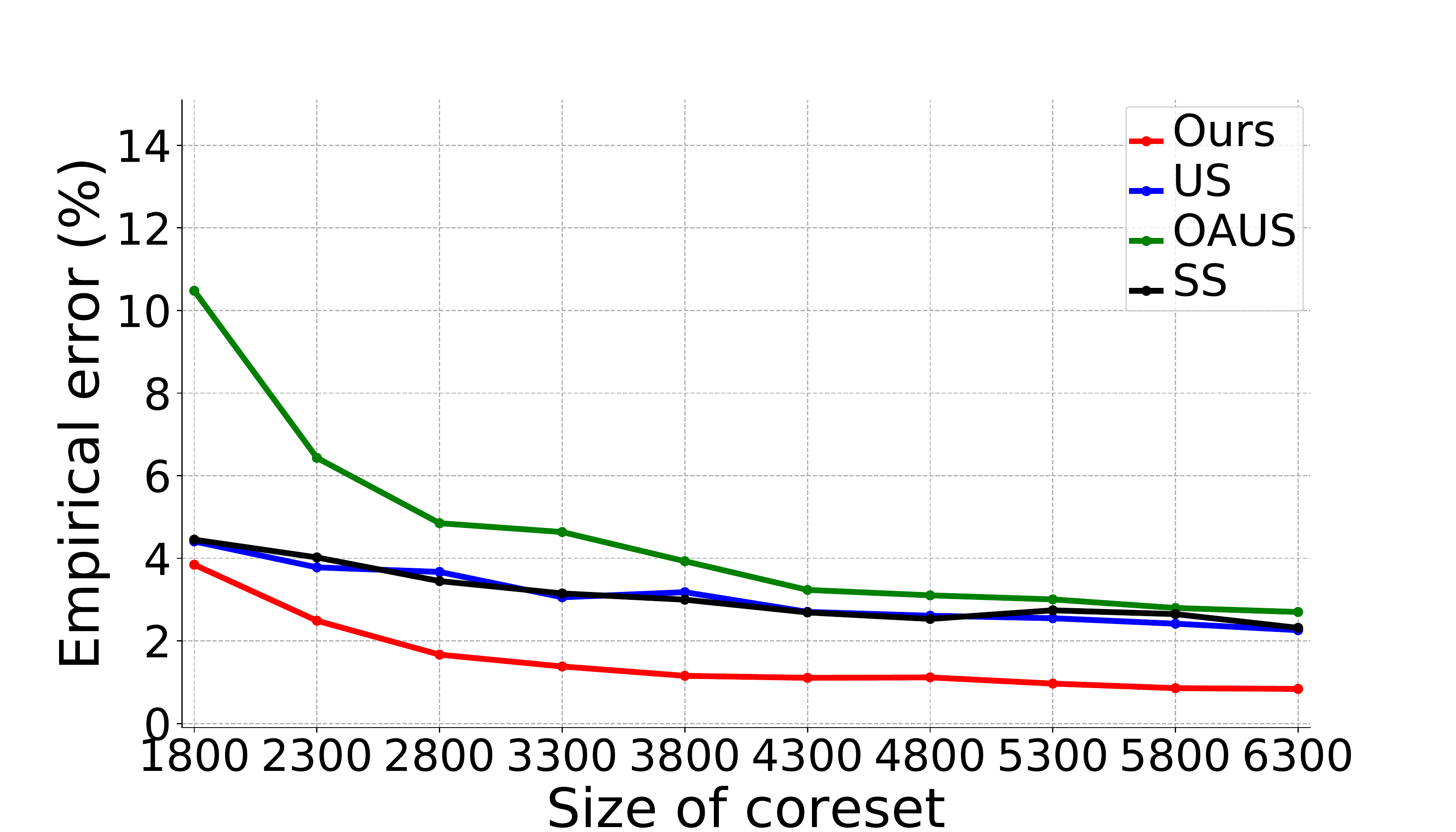}
        \caption*{Census1990 dataset}
    \end{subfigure}
    \caption{The tradeoff between the coreset size and the empirical error.}
    \label{fig:size_vs_error}
\end{figure}

\paragraph{Experiment: Impact of The Number of Outliers}
We also examine the impact of the number of outliers $m$ on empirical error.
Specifically, we experiment with varying $m$, but a fixed $N-m$,
which is the number of ``samples'' besides the included outliers $L^*$ in our algorithm.
We pick a typical value of $N - m = 800$ based on the curves of \Cref{fig:size_vs_error}, .
We plot this outlier-error curve in \Cref{fig:m_vs_error},
and we observe that while some of our baselines have a fluctuating empirical error,
the error curve of our coreset is relatively stable.
This suggests that the empirical error of our coreset is mainly determined by the
number of additional samples $N - m$, and is mostly independent of the number of outliers $m$ itself.
\begin{figure}[t]
    \centering
    \begin{subfigure}[b]{0.42\textwidth}
        \centering
        \includegraphics[width=\textwidth]{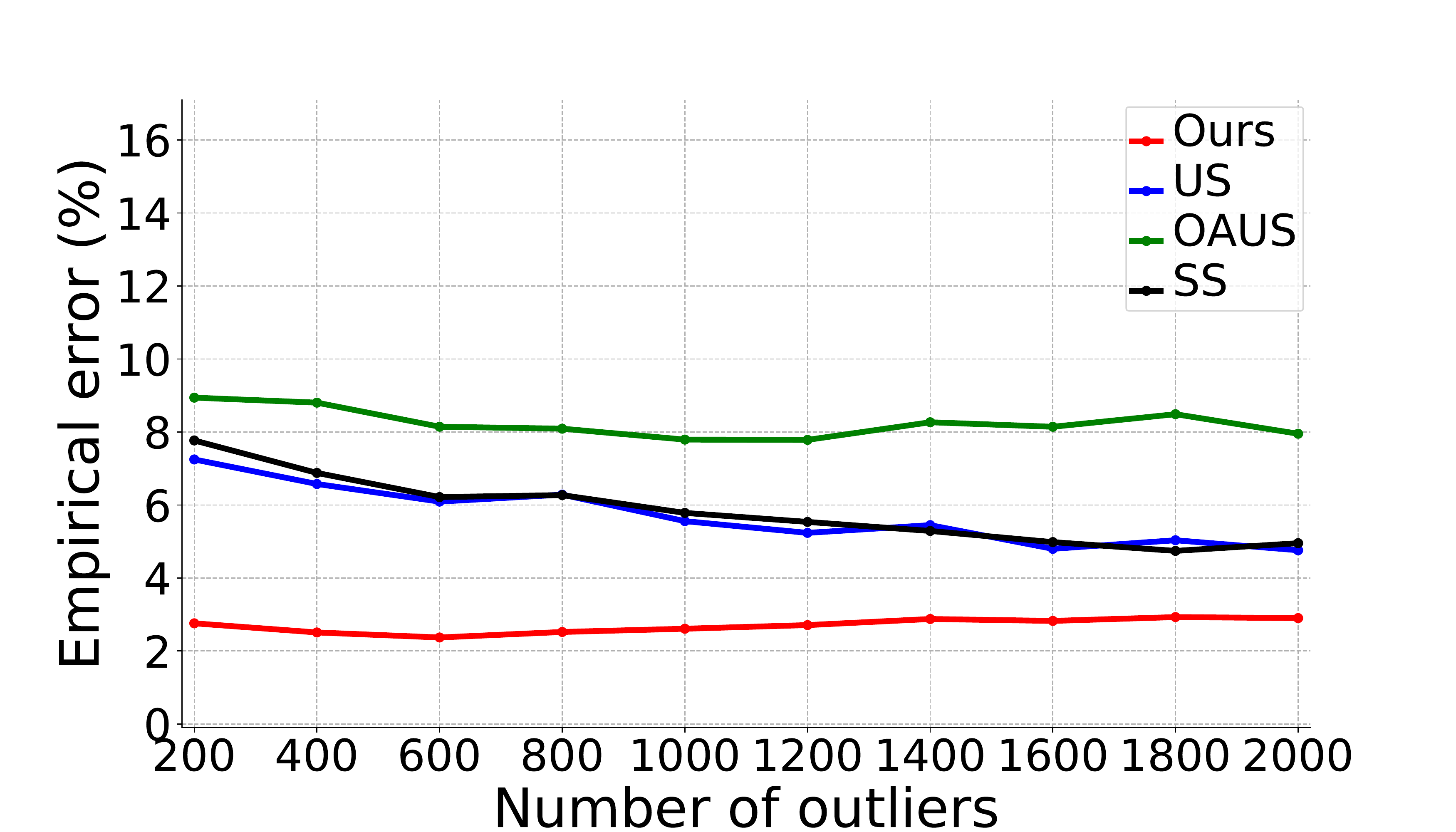}
        \caption*{Adult dataset}
    \end{subfigure} \qquad
    \begin{subfigure}[b]{0.42\textwidth}
        \centering
        \includegraphics[width=\textwidth]{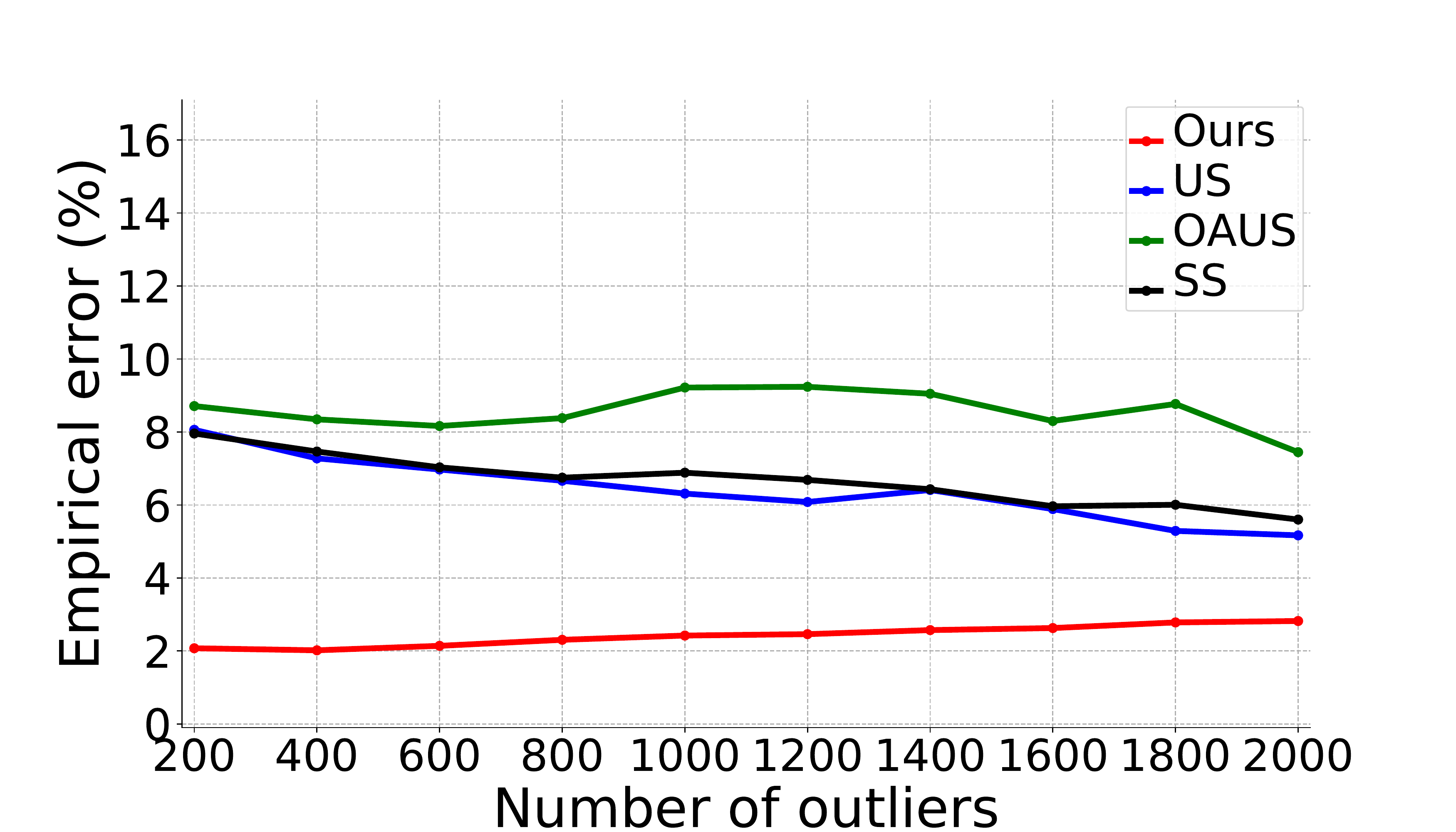}
        \caption*{Bank dataset}
    \end{subfigure}
    \begin{subfigure}[b]{0.42\textwidth}
        \centering
        \includegraphics[width=\textwidth]{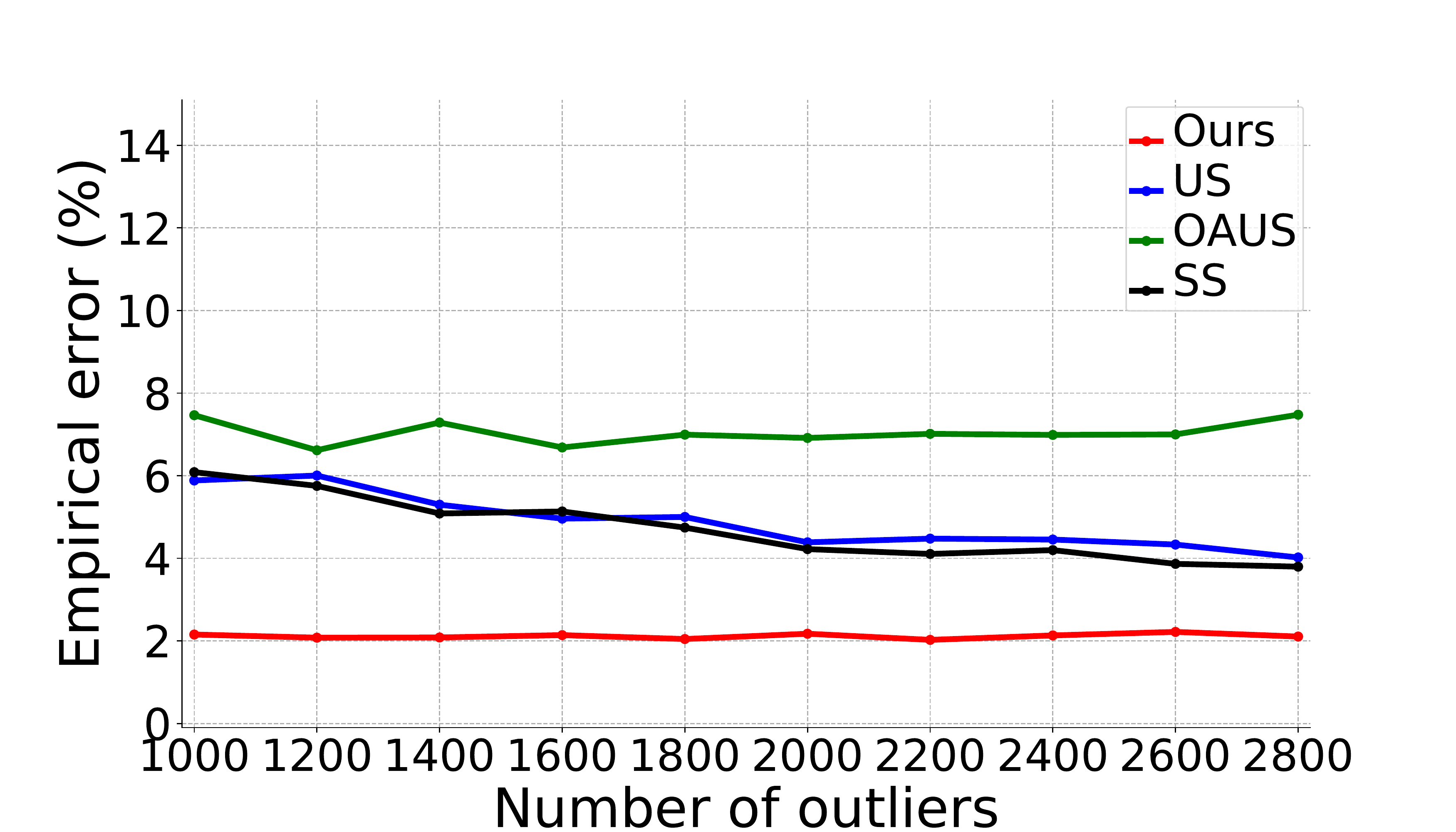}
        \caption*{Twitter dataset}
    \end{subfigure} \qquad
    \begin{subfigure}[b]{0.42\textwidth}
        \centering
        \includegraphics[width=\textwidth]{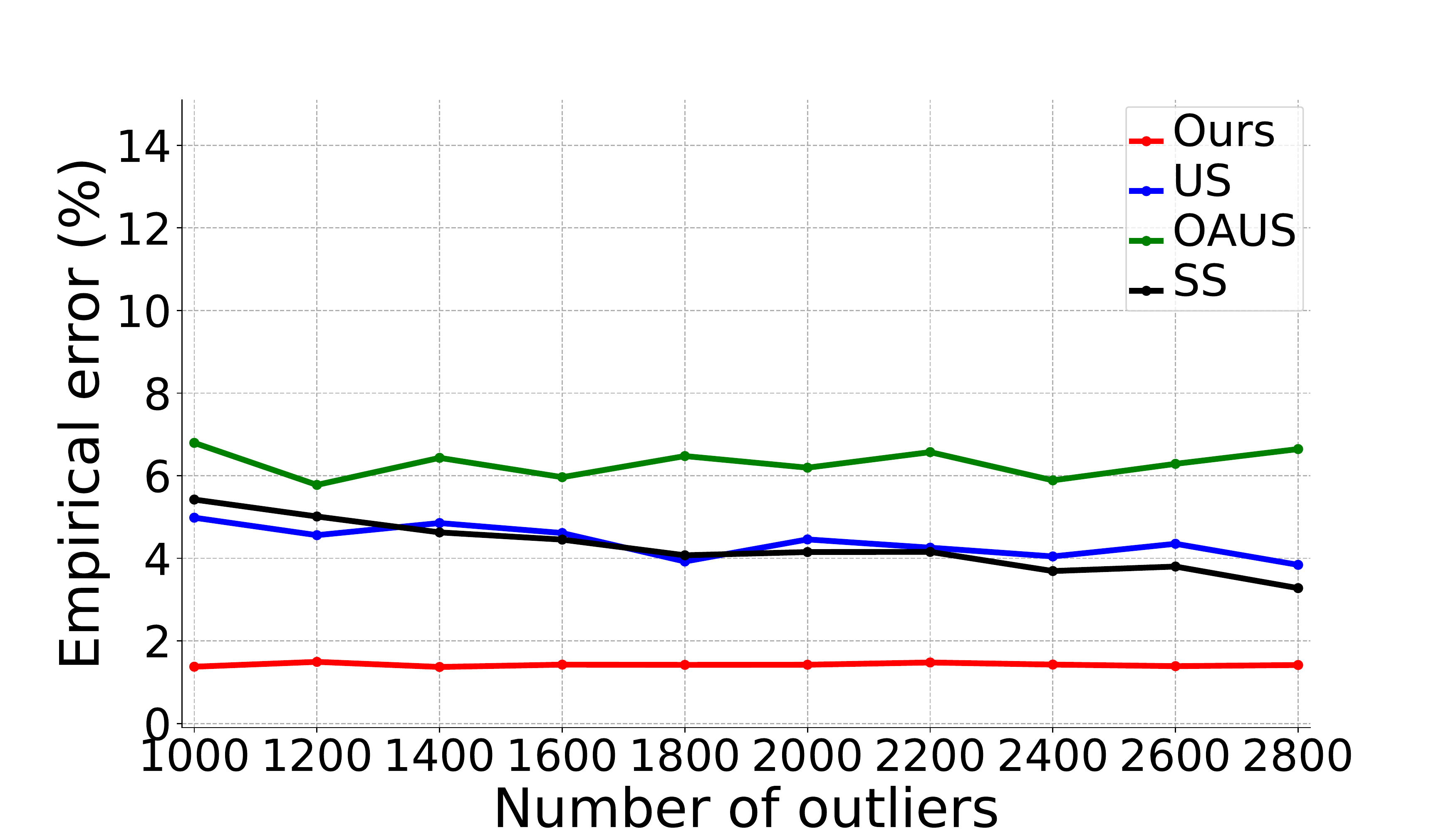}
        \caption*{Census1990 dataset}
    \end{subfigure}
    \caption{The impact of the number of outliers $m$ on the empirical error}
    \label{fig:m_vs_error}
\end{figure}

\subsection{Speeding Up Existing Approximation Algorithms}
We validate the ability of our coresets for speeding up existing approximation algorithms for robust clustering.
We consider two natural algorithms and run them on top of our coreset for speedup:
a Lloyd-style algorithm tailored to \kmMeans~\cite{chawla2013kmeans}
seeded by a modified \kMeanspp for robust clustering~\cite{bhaskara2019greedy}, which we call ``LL'',
and a local search algorithm for \kmMedian~\cite{friggstad2019approximation}, which we call ``LS''.
We note that for LS, we uniformly sample $100$ points from the dataset and use them as the only potential centers,
since otherwise it takes too long to run on the original dataset (without coresets).
We use a coreset of size $m + 500$ for each dataset
(recalling that $m$ is picked per dataset according to \Cref{fig:outliers})
to speed up the algorithms.
To make a consistent comparison, we measure the clustering costs on the original dataset for all runs (instead of on the coreset).

We report in \Cref{tab:speedup} the running time and the cost achieved by LL and LS, with and without coresets.
The results show that the error incurred by using coreset is tiny ($<5\%$ error),
but the speedup is a significant 80x-250x for LL, and a 100x-200x for LS.
Even taking the coreset construction time into consideration, it still achieves a 10x-30x speedup to LL and a 80x-140x speedup to LS.
We conclude that our coreset drastically improves the running time for existing approximation algorithms, while only suffering a neglectable error.

\begin{table}[t]
	\centering
	\captionsetup{font=small}
	\small
	\caption{Running time and costs for LL and LS with/without coresets.
		$T_X$ and $T_S$ are the running time without/with the coreset, respectively.
		Similarly, $\cost$ and $\cost'$ are the clustering costs without/with the coreset. $T_C$ is coreset construction time.
This entire experiment is repeated $10$ times and the average is reported.
	}
	\label{tab:speedup}
	\begin{tabular}{cccccccc}
		\toprule
		dataset & algorithm & $\cost$ & $\cost'$  & $T_{C}$ (s) & $T_{S}$ (s) & $T_{X}$ (s) \\
		\midrule
		\multirow{2}{*}{{Adult}} &
		LL & $3.790\times 10^{13}$ & $3.922\times 10^{13}$ & $0.4657$ & $0.06385$ &  $16.51$  \\
		& LS & $1.100\times 10^9$ & $1.107\times 10^9$ & $0.5300$  & $1.147$ & $204.8$  \\
		\midrule
		\multirow{2}{*}{{Bank}} &
		LL  & $4.444\times 10^8$ & $4.652\times 10^8$ & $0.4399$ & $0.05900$ & $11.40$  \\
		& LS  & $4.717\times 10^6$ & $4.721\times 10^6$ & $0.4953$  & $1.220$ & $186.6$  \\
		\midrule
		\multirow{2}{*}{{Twitter}} &
		LL   & $3.218\times 10^7$ & $3.236\times 10^7$ & $0.9493$ & $0.08289$ &  $11.27$  \\
		& LS  & $1.476\times 10^6$ & $1.451 \times 10^6$ & $1.064$  & $2.135$ & $460.2$  \\
		\midrule
		\multirow{2}{*}{{Census1990}} &
		LL   & $1.189 \times 10^7$ & $1.208\times 10^7$ & $3.673$ & $0.4809$ & $40.54$ \\
		& LS  & $1.165 \times 10^6$ & $1.163 \times 10^6$ & $4.079$  & $24.83$ & $2405$  \\
		\bottomrule
	\end{tabular}
	
\end{table} 
\bibliographystyle{alpha}
\bibliography{ref}

\newcommand{\etalchar}[1]{$^{#1}$}
\begin{thebibliography}{BCAJ{\etalchar{+}}22}

\bibitem[BBH{\etalchar{+}}20]{baker2020coresets}
Daniel~N. Baker, Vladimir Braverman, Lingxiao Huang, Shaofeng~H.{-}C. Jiang,
  Robert Krauthgamer, and Xuan Wu.
\newblock Coresets for clustering in graphs of bounded treewidth.
\newblock In {\em {ICML}}, volume 119 of {\em Proceedings of Machine Learning
  Research}, pages 569--579. {PMLR}, 2020.

\bibitem[BCAJ{\etalchar{+}}22]{braverman2022power}
Vladimir Braverman, Vincent Cohen-Addad, Shaofeng Jiang, Robert Krauthgamer,
  Chris Schwiegelshohn, Mads~Bech Toftrup, and Xuan Wu.
\newblock The power of uniform sampling for coresets.
\newblock In {\em {FOCS} (to appear)}. {IEEE} Computer Society, 2022.

\bibitem[BEL13]{BalcanEL13}
Maria{-}Florina Balcan, Steven Ehrlich, and Yingyu Liang.
\newblock Distributed $k$-means and $k$-median clustering on general
  communication topologies.
\newblock In {\em {NIPS}}, pages 1995--2003, 2013.

\bibitem[BFL16]{braverman2016new}
Vladimir Braverman, Dan Feldman, and Harry Lang.
\newblock New frameworks for offline and streaming coreset constructions.
\newblock {\em CoRR}, abs/1612.00889, 2016.

\bibitem[BJKW19]{braverman2019coresets}
Vladimir Braverman, Shaofeng~H.{-}C. Jiang, Robert Krauthgamer, and Xuan Wu.
\newblock Coresets for ordered weighted clustering.
\newblock In {\em {ICML}}, volume~97 of {\em Proceedings of Machine Learning
  Research}, pages 744--753. {PMLR}, 2019.

\bibitem[BJKW21]{BJKW21}
Vladimir Braverman, Shaofeng~H.{-}C. Jiang, Robert Krauthgamer, and Xuan Wu.
\newblock Coresets for clustering in excluded-minor graphs and beyond.
\newblock In {\em {SODA}}, pages 2679--2696. {SIAM}, 2021.

\bibitem[BK18]{bhaskara2018low}
Aditya Bhaskara and Srivatsan Kumar.
\newblock Low rank approximation in the presence of outliers.
\newblock In {\em {APPROX-RANDOM}}, volume 116 of {\em LIPIcs}, pages
  4:1--4:16. Schloss Dagstuhl - Leibniz-Zentrum f{\"{u}}r Informatik, 2018.

\bibitem[BVX19]{bhaskara2019greedy}
Aditya Bhaskara, Sharvaree Vadgama, and Hong Xu.
\newblock Greedy sampling for approximate clustering in the presence of
  outliers.
\newblock In {\em NeurIPS}, pages 11146--11155, 2019.

\bibitem[CG13]{chawla2013kmeans}
Sanjay Chawla and Aristides Gionis.
\newblock k-means{-}{-}: {A} unified approach to clustering and outlier
  detection.
\newblock In {\em {SDM}}, pages 189--197. {SIAM}, 2013.

\bibitem[CGS18]{twitter_data}
T-H.~Hubert Chan, Arnaud Guerquin, and Mauro Sozio.
\newblock Twitter data set, 2018.
\newblock \url{https://github.com/fe6Bc5R4JvLkFkSeExHM/k-center}.

\bibitem[Che08]{chen2008constant}
Ke~Chen.
\newblock A constant factor approximation algorithm for \emph{k}-median
  clustering with outliers.
\newblock In {\em {SODA}}, pages 826--835. {SIAM}, 2008.

\bibitem[CKMN01]{charikar2001algorithms}
Moses Charikar, Samir Khuller, David~M. Mount, and Giri Narasimhan.
\newblock Algorithms for facility location problems with outliers.
\newblock In {\em {SODA}}, pages 642--651. {ACM/SIAM}, 2001.

\bibitem[CLMW11]{Cands2011RobustPC}
Emmanuel~J. Cand{\`e}s, Xiaodong Li, Yi~Ma, and John Wright.
\newblock Robust principal component analysis?
\newblock {\em J. ACM}, 58:11:1--11:37, 2011.

\bibitem[CLSS22]{cohen2022towards}
Vincent Cohen{-}Addad, Kasper~Green Larsen, David Saulpic, and Chris
  Schwiegelshohn.
\newblock Towards optimal lower bounds for $k$-median and $k$-means coresets.
\newblock In {\em {STOC}}, pages 1038--1051. {ACM}, 2022.

\bibitem[CSS21a]{cohen-addad2021improved}
Vincent Cohen{-}Addad, David Saulpic, and Chris Schwiegelshohn.
\newblock Improved coresets and sublinear algorithms for power means in
  euclidean spaces.
\newblock In {\em NeurIPS}, pages 21085--21098, 2021.

\bibitem[CSS21b]{cohen2021new}
Vincent Cohen{-}Addad, David Saulpic, and Chris Schwiegelshohn.
\newblock A new coreset framework for clustering.
\newblock In {\em {STOC}}, pages 169--182. {ACM}, 2021.

\bibitem[DG17]{adult_data}
Dheeru Dua and Casey Graff.
\newblock Uci machine learning repository, 2017.
\newblock \url{https://archive.ics.uci.edu/ml/datasets/adult}.

\bibitem[DKP20]{deshpande2020robust}
Amit Deshpande, Praneeth Kacham, and Rameshwar Pratap.
\newblock Robust k-means++.
\newblock In {\em {UAI}}, volume 124 of {\em Proceedings of Machine Learning
  Research}, pages 799--808. {AUAI} Press, 2020.

\bibitem[DW20]{DBLP:conf/icml/DingW20}
Hu~Ding and Zixiu Wang.
\newblock Layered sampling for robust optimization problems.
\newblock In {\em {ICML}}, volume 119 of {\em Proceedings of Machine Learning
  Research}, pages 2556--2566. {PMLR}, 2020.

\bibitem[FKRS19]{friggstad2019approximation}
Zachary Friggstad, Kamyar Khodamoradi, Mohsen Rezapour, and Mohammad~R.
  Salavatipour.
\newblock Approximation schemes for clustering with outliers.
\newblock {\em {ACM} Trans. Algorithms}, 15(2):26:1--26:26, 2019.

\bibitem[FL11]{feldman2011unified}
Dan Feldman and Michael Langberg.
\newblock A unified framework for approximating and clustering data.
\newblock In {\em {STOC}}, pages 569--578. {ACM}, 2011.
\newblock \url{https://arxiv.org/abs/1106.1379}.

\bibitem[FS12]{feldman2012data}
Dan Feldman and Leonard~J Schulman.
\newblock Data reduction for weighted and outlier-resistant clustering.
\newblock In {\em Proceedings of the twenty-third annual ACM-SIAM symposium on
  Discrete Algorithms}, pages 1343--1354. SIAM, 2012.

\bibitem[FSS20]{FSS20}
Dan Feldman, Melanie Schmidt, and Christian Sohler.
\newblock Turning big data into tiny data: Constant-size coresets for
  $k$-means, pca, and projective clustering.
\newblock {\em {SIAM} J. Comput.}, 49(3):601--657, 2020.

\bibitem[FZH{\etalchar{+}}19]{feng2019improved}
Qilong Feng, Zhen Zhang, Ziyun Huang, Jinhui Xu, and Jianxin Wang.
\newblock Improved algorithms for clustering with outliers.
\newblock In {\em {ISAAC}}, volume 149 of {\em LIPIcs}, pages 61:1--61:12.
  Schloss Dagstuhl - Leibniz-Zentrum f{\"{u}}r Informatik, 2019.

\bibitem[GKL{\etalchar{+}}17]{gupta2017local}
Shalmoli Gupta, Ravi Kumar, Kefu Lu, Benjamin Moseley, and Sergei
  Vassilvitskii.
\newblock Local search methods for $k$-means with outliers.
\newblock {\em Proc. {VLDB} Endow.}, 10(7):757--768, 2017.

\bibitem[HJLW18]{huang2018epsilon}
Lingxiao Huang, Shaofeng~H.{-}C. Jiang, Jian Li, and Xuan Wu.
\newblock Epsilon-coresets for clustering (with outliers) in doubling metrics.
\newblock In {\em {FOCS}}, pages 814--825. {IEEE} Computer Society, 2018.

\bibitem[HJV19]{huang2019coresets}
Lingxiao Huang, Shaofeng~H.{-}C. Jiang, and Nisheeth~K. Vishnoi.
\newblock Coresets for clustering with fairness constraints.
\newblock In {\em NeurIPS}, pages 7587--7598, 2019.

\bibitem[HK20]{DBLP:conf/esa/HenzingerK20}
Monika Henzinger and Sagar Kale.
\newblock Fully-dynamic coresets.
\newblock In {\em {ESA}}, volume 173 of {\em LIPIcs}, pages 57:1--57:21.
  Schloss Dagstuhl - Leibniz-Zentrum f{\"{u}}r Informatik, 2020.

\bibitem[HM04]{DBLP:conf/stoc/Har-PeledM04}
Sariel Har{-}Peled and Soham Mazumdar.
\newblock On coresets for $k$-means and $k$-median clustering.
\newblock In {\em {STOC}}, pages 291--300. {ACM}, 2004.
\newblock \url{https://arxiv.org/abs/1810.12826}.

\bibitem[Hp11]{GAA}
Sariel Har-peled.
\newblock {\em Geometric Approximation Algorithms}.
\newblock American Mathematical Society, USA, 2011.

\bibitem[HR09]{huber2009robust}
Peter~J Huber and Elvezio Ronchetti.
\newblock {\em Robust Statistics}.
\newblock Wiley, 2009.

\bibitem[HV20]{huang2020coresets}
Lingxiao Huang and Nisheeth~K. Vishnoi.
\newblock Coresets for clustering in {Euclidean} spaces: importance sampling is
  nearly optimal.
\newblock In {\em {STOC}}, pages 1416--1429. {ACM}, 2020.

\bibitem[KLS18]{RLS18}
Ravishankar Krishnaswamy, Shi Li, and Sai Sandeep.
\newblock Constant approximation for $k$-median and $k$-means with outliers via
  iterative rounding.
\newblock In {\em {STOC}}, pages 646--659. {ACM}, 2018.

\bibitem[LLS01]{LI2001516}
Yi~Li, Philip~M. Long, and Aravind Srinivasan.
\newblock Improved bounds on the sample complexity of learning.
\newblock {\em J. Comput. Syst. Sci.}, 62(3):516--527, 2001.

\bibitem[MCR14]{bank_data}
S.~Moro, P.~Cortez, and P.~Rita.
\newblock Uci machine learning repository, 2014.
\newblock \url{https://archive.ics.uci.edu/ml/datasets/Bank+Marketing}.

\bibitem[MMR19]{MMR19}
Konstantin Makarychev, Yury Makarychev, and Ilya~P. Razenshteyn.
\newblock Performance of johnson-lindenstrauss transform for \emph{k}-means and
  \emph{k}-medians clustering.
\newblock In {\em {STOC}}, pages 1027--1038. {ACM}, 2019.

\bibitem[MNP{\etalchar{+}}14]{mount2014on}
David~M. Mount, Nathan~S. Netanyahu, Christine~D. Piatko, Ruth Silverman, and
  Angela~Y. Wu.
\newblock On the least trimmed squares estimator.
\newblock {\em Algorithmica}, 69(1):148--183, 2014.

\bibitem[MTH90]{census1990_data}
Chris Meek, Bo~Thiesson, and David Heckerman.
\newblock Uci machine learning repository, 1990.
\newblock \url{http://archive.ics.uci.edu/ml/datasets/US+Census+Data+(1990)}.

\bibitem[NN19]{NN19}
Shyam Narayanan and Jelani Nelson.
\newblock Optimal terminal dimensionality reduction in euclidean space.
\newblock In {\em {STOC}}, pages 1064--1069. {ACM}, 2019.

\bibitem[RL87]{rousseeuw1987robust}
Peter~J. Rousseeuw and Annick Leroy.
\newblock {\em Robust Regression and Outlier Detection}.
\newblock Wiley Series in Probability and Statistics. Wiley, 1987.

\bibitem[SRF20]{statman2020kmeans}
Adiel Statman, Liat Rozenberg, and Dan Feldman.
\newblock $k$-means: Outliers-resistant clustering+++.
\newblock {\em Algorithms}, 13(12):311, 2020.

\bibitem[SSS19]{schmidt2019fair}
Melanie Schmidt, Chris Schwiegelshohn, and Christian Sohler.
\newblock Fair coresets and streaming algorithms for fair k-means.
\newblock In {\em {WAOA}}, volume 11926 of {\em Lecture Notes in Computer
  Science}, pages 232--251. Springer, 2019.

\bibitem[SW18]{sohler2018strong}
Christian Sohler and David~P. Woodruff.
\newblock Strong coresets for $k$-median and subspace approximation: Goodbye
  dimension.
\newblock In {\em {FOCS}}, pages 802--813. {IEEE} Computer Society, 2018.

\bibitem[ZFH{\etalchar{+}}21]{zhang2021local}
Zhen Zhang, Qilong Feng, Junyu Huang, Yutian Guo, Jinhui Xu, and Jianxin Wang.
\newblock A local search algorithm for $k$-means with outliers.
\newblock {\em Neurocomputing}, 450:230--241, 2021.

\end{thebibliography}

\begin{appendices}
    \appendixpage
    \section{Algorithms for Tri-criteria Approximation}
    \label{sec:tri_criteria}
    Various known algorithms that offer different tradeoffs may be used for the required $(\alpha,\beta,\gamma)$-approximation.
    In particular, \cite{friggstad2019approximation} designed a polynomial-time
    $\mathcal{A}(n,k,d,z)=n^{O(1)}$ algorithm with $\alpha=O(2^z),\beta=O(1),\gamma=1$;
    \cite{bhaskara2019greedy} gave a near-linear time $\mathcal{A}(n,k,d,z)=\tilde{O}(nkd)$
    algorithm with $\alpha=O(2^{O(z)}),\beta = O(1),\gamma=O(1)$ (which implies the statement in \Cref{thm:intro_main}).\footnote{\cite{bhaskara2019greedy} only showed the case of $z = 2$, but we check that it also generalizes to other $z$'s.}
    Finally, true approximation algorithms, i.e., $\beta = \gamma = 1$,
    are known for both \kmMedian and \kmMeans, and they run in polynomial-time
    $\mathcal{A}(n,k,d,z)=n^{O(1)}$ and achieves $\alpha=O(1)$~\cite{chen2008constant,RLS18}.

\section{Determining The Number of Outliers in The Experiments}
\label{sec:determine_outlier}
To determine the number of outliers $m$ for each (subsampled) dataset in our experiment,
we run a vanilla \kMeans clustering (without outliers) algorithm,
and plot the distribution of distances from data points to the found near-optimal centers.
As shown in \Cref{fig:outliers}, every dataset admits a clear breaking point that defines the outliers, and we pick $m$ accordingly.

\begin{figure}[t]
    \centering
    \begin{subfigure}[b]{0.42\textwidth}
        \centering
        \includegraphics[width=\textwidth]{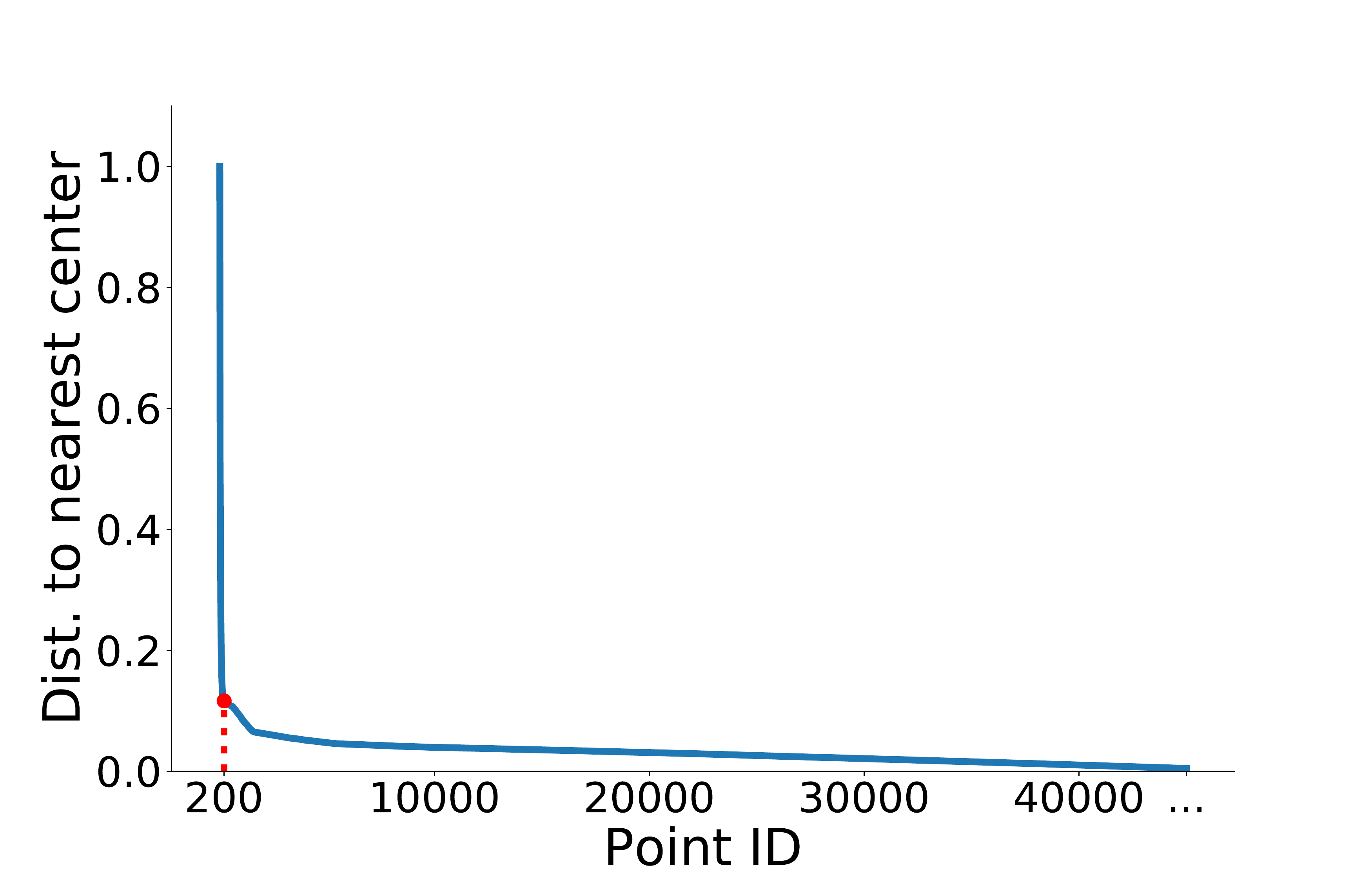}
        \caption*{Adult dataset}
    \end{subfigure} \qquad
    \begin{subfigure}[b]{0.42\textwidth}
        \centering
        \includegraphics[width=\textwidth]{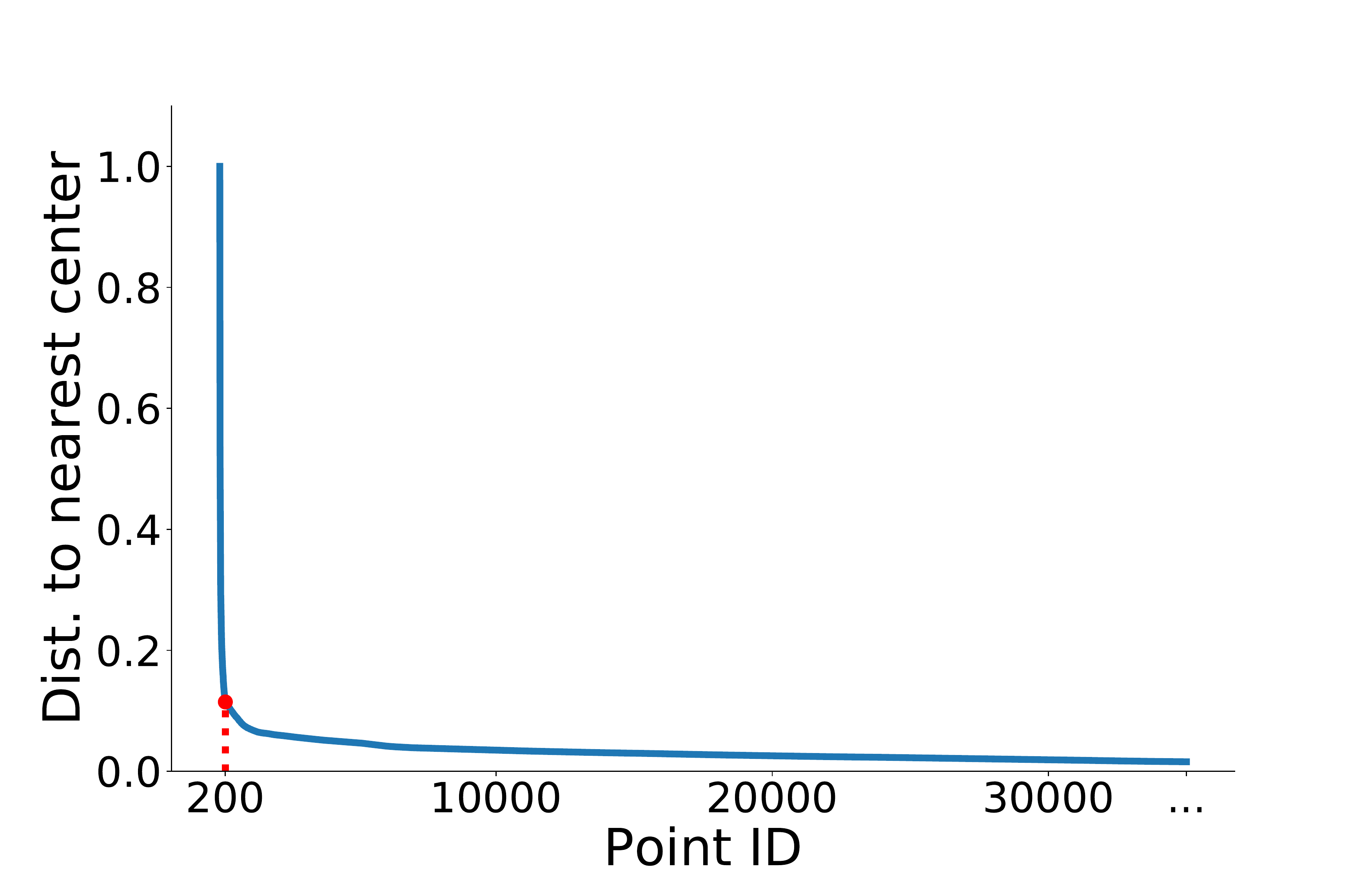}
        \caption*{Bank dataset}
    \end{subfigure}
    \begin{subfigure}[b]{0.42\textwidth}
        \centering
        \includegraphics[width=\textwidth]{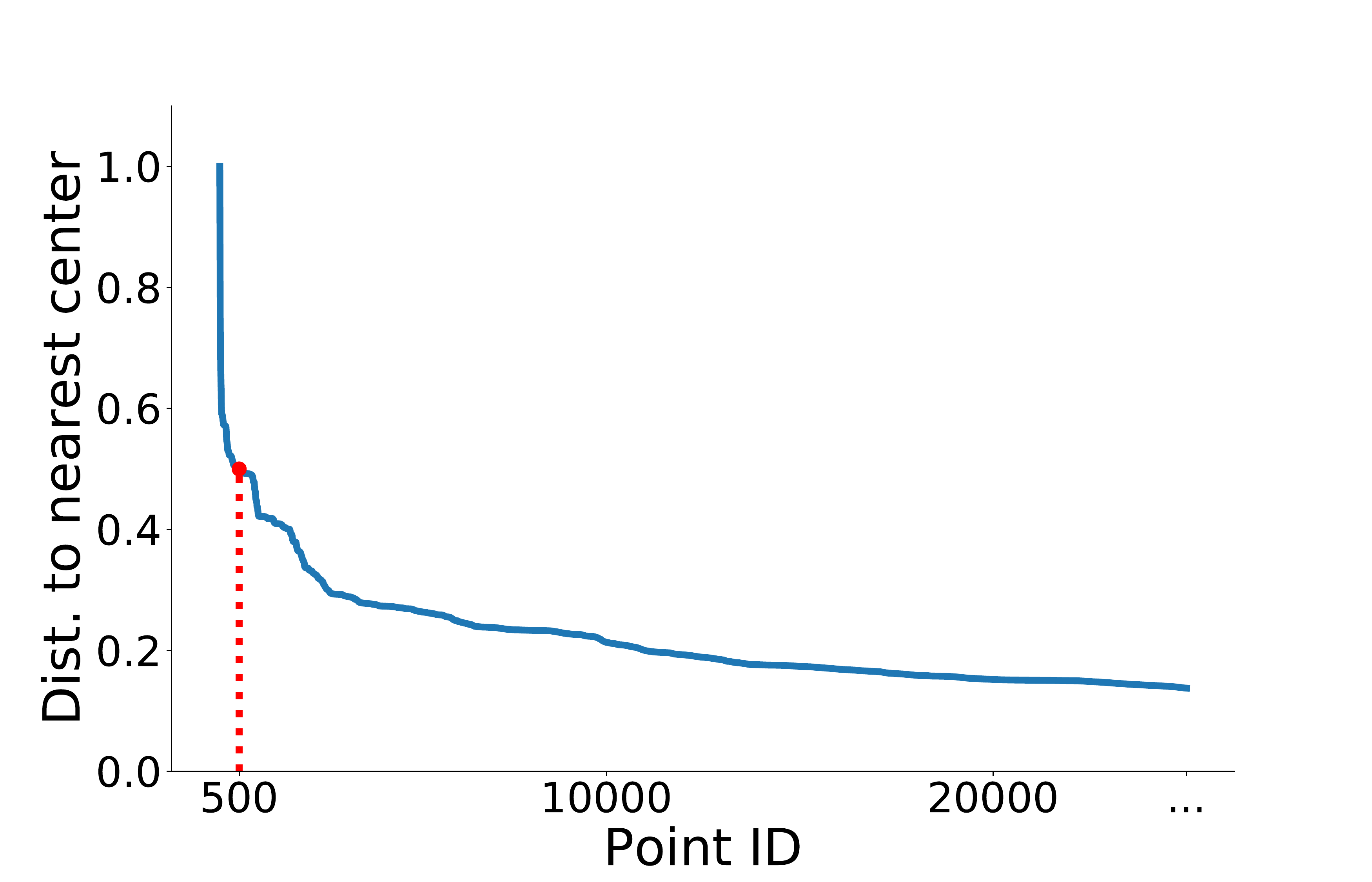}
        \caption*{Twitter dataset}
    \end{subfigure} \qquad
    \begin{subfigure}[b]{0.42\textwidth}
        \centering
        \includegraphics[width=\textwidth]{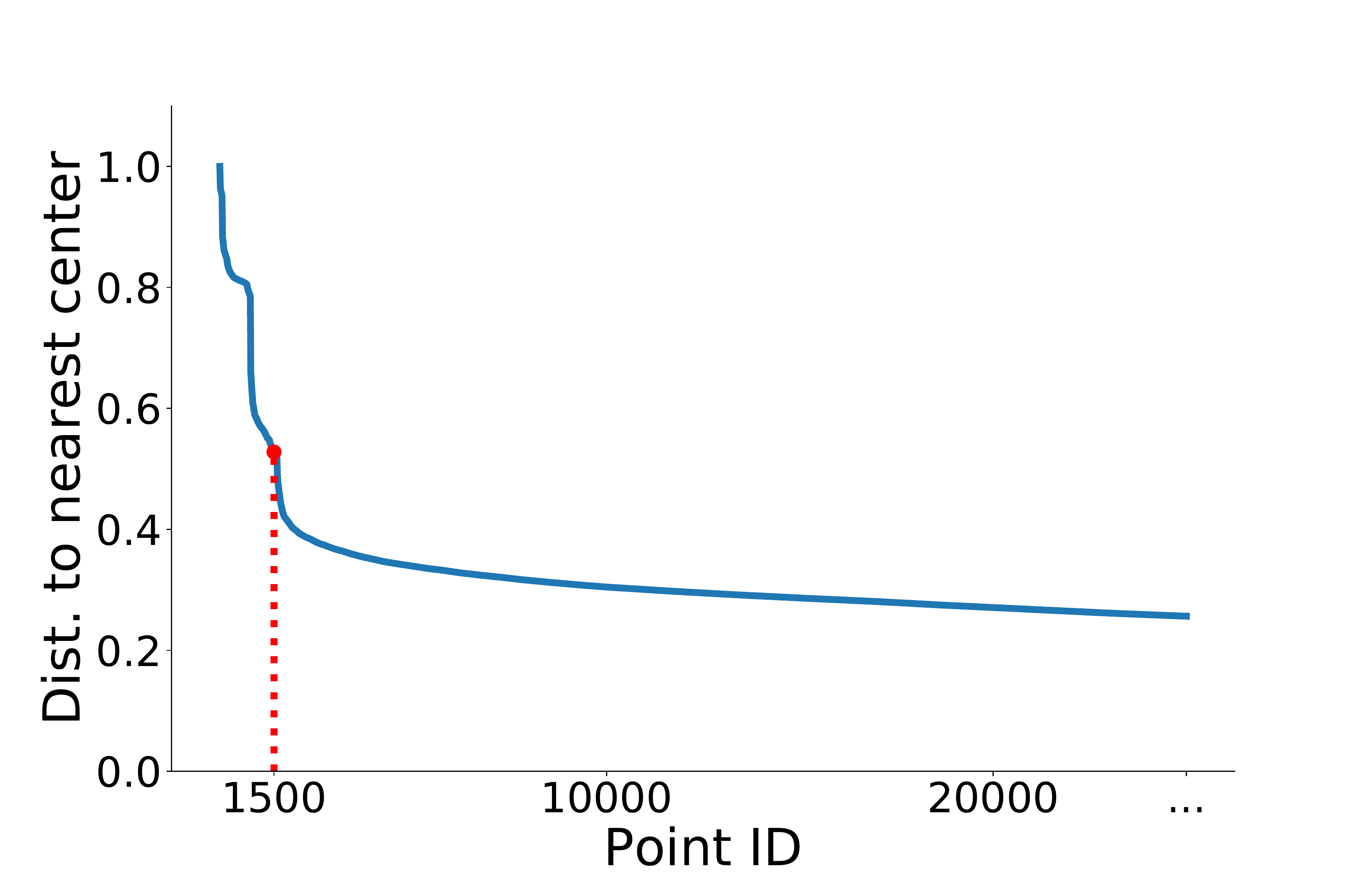}
        \caption*{Census1990 dataset}
    \end{subfigure}
    \caption{Distances to the found near-optimal center (using a vanilla clustering algorithm) for each point, sorted decreasingly and rescaled to $[0, 1]$.}
    \label{fig:outliers}
\end{figure}

 \end{appendices}

\end{document}